\documentclass[%
 reprint,
superscriptaddress,
 amsmath,amssymb,
 aps,
prl,
]{revtex4-2}

\usepackage{graphicx}% Include figure files
\usepackage{dcolumn}% Align table columns on decimal point
\usepackage{bm}% bold math
\usepackage[colorlinks=true,linkcolor=blue,citecolor=blue,urlcolor=blue]{hyperref}% add hypertext capabilities
\usepackage{times}
\usepackage{amsmath}
\usepackage{amssymb}
\usepackage{stackrel}
\usepackage{tikz-cd}
\usepackage{tikz}
\usepackage[cm]{fullpage}
\usepackage{amsthm} 
\usepackage{mathtools}
\usepackage{enumitem}
\usepackage{bbold}
\usepackage[normalem]{ulem}
\usepackage{hyperref}
\usepackage{cleveref}
\usepackage{appendix}
\usepackage{comment}
\usepackage{subfigure}
\usepackage{centernot}
\usepackage{braket}
\usepackage{algorithm}
\usepackage{algpseudocode}
\usepackage{varwidth}
\usepackage{ulem}
\usepackage{apptools}
\AtAppendix{\counterwithin{Lemma}{section}}
\AtAppendix{\counterwithin{Theorem}{section}}
\AtAppendix{\counterwithin{Proposition}{section}}

\newtheorem{Theorem}{Theorem}
\newtheorem{Lemma}{Lemma}
\newtheorem{Proposition}{Proposition}
\newtheorem{Corollary}{Corollary}

\theoremstyle{definition}

\definecolor{fuchsia}{rgb}{1.0, 0.0, 1.0}

\begin{document}

\preprint{APS/123-QED}

\title{Optimally generating $\mathfrak{su}(2^N)$ using Pauli strings}

\author{Isaac D. Smith}
\email{isaac.smith@uibk.ac.at}
\affiliation{University of Innsbruck, Institute for Theoretical Physics, Technikerstr. 21A, Innsbruck A-6020, Austria}
\author{Maxime Cautrès}
\affiliation{University of Innsbruck, Institute for Theoretical Physics, Technikerstr. 21A, Innsbruck A-6020, Austria}
\affiliation{École Normale Supérieure de Lyon, F-69007 Lyon, France} 
\author{David T. Stephen}
\affiliation{Department of Physics and Center for Theory of Quantum Matter, University of Colorado Boulder, Boulder, Colorado 80309 USA}
\affiliation{Department of Physics and Institute for Quantum Information and Matter, California Institute of Technology, Pasadena, California 91125, USA}
\author{Hendrik Poulsen Nautrup}
\affiliation{University of Innsbruck, Institute for Theoretical Physics, Technikerstr. 21A, Innsbruck A-6020, Austria}

\date{\today}

\begin{abstract} Any quantum computation consists of a sequence of unitary evolutions described by a finite set of Hamiltonians. When this set is taken to consist of only products of Pauli operators, we show that the minimal such set generating $\mathfrak{su}(2^{N})$ contains $2N+1$ elements. We provide a number of examples of such generating sets and furthermore provide an algorithm for producing a sequence of rotations corresponding to any given Pauli rotation, which is shown to have optimal complexity. We also observe that certain sets generate $\mathfrak{su}(2^{N})$ at a faster rate than others, and we show how this rate can be optimized by tuning the fraction of anticommuting pairs of generators. Finally, we briefly comment on implications for measurement-based and trapped ion quantum computation as well as the construction of fault-tolerant gate sets.
\end{abstract}

\maketitle

Performing a quantum computation requires the control of a quantum system. The study of the controllability of quantum systems has a substantial history with developments across both mathematics and physics \cite{huang1983controllability,jurdjevic1972control,Kuranishi1951On,Ramakrishna,albertini2001notions,albertini2002lie,Schirmer_01,d2007introduction,Zeier2011}. The techniques typically applied in control theory are often drawn from algebra \cite{sussmann1972controllability}, with much focus on Lie algebra in the quantum context \cite{huang1983controllability,d2007introduction}. 

In quantum computing, control refers to the intentional evolution of a closed quantum system from a known initial state to a state representing the logical output of the computation \cite{ramakrishna1996relation,Lloyd_95}. As the evolution of a closed quantum system is described by an element of the (special) unitary group, the corresponding Lie algebra (and subalgebras thereof), whose elements generate unitary evolution, have received much interest. As they describe system dynamics, these Lie algebras are often called \textit{dynamical Lie algebras} (DLAs) \cite{d2007introduction,d2021dynamical,wiersema2023classification}.

DLAs, like any Lie algebra, are vector spaces equipped with an associated Lie bracket (e.g., matrix commutator). Typically, the DLA of a quantum system is \textit{generated} from a smaller set $\mathcal{A}$ of ``elementary operations'' using nested commutators and linear combinations. If this DLA describes all elements of the special unitary group, then $\mathcal{A}$ is a universal generating set. In quantum computing, $\mathcal{A}$ corresponds to the elementary gates, while the corresponding DLA describes the set of unitaries that can be implemented by gate composition. Significant past work has aimed to characterize such universal sets for various quantum systems \cite{Lloyd_95,Schirmer_01,schirmer2002identification,d2007introduction,Ramakrishna,Zeier2011}.

Here, we aim to understand how small a generating set $\mathcal{A}$ can be, while still maintaining universality. If no further restriction is made on $\mathcal{A}$, then the answer is known: the minimal generating set for $\mathfrak{su}(2^N)$ has only $2$ elements (see Theorem 6 in Ref.~\cite{Kuranishi1951On}). However, the elements of this set consist of large sums of Pauli operators and are not particularly practical from a quantum information theory perspective. A more natural generating set would consist of only tensor products of Pauli operators (Pauli strings). Generating sets such as these occur in computational schemes related to measurement-based quantum computation (MBQC) \cite{Raussendorf2003,Raussendorf2019,Devakul2018,Stephen2019subsystem,Daniel2020computational,Parity_MBQC,Stephen2024} and quantum cellular automata \cite{Raussendorf2005,nautrup2023}. Moreover, rotations of Pauli strings arise in simulation of fermionic systems on qubit devices~\cite{steudtner2018fermion,romero2018strategies,cowtan2020generic,van2020circuit} and are implemented in the native gate sets of e.g., trapped ion systems~\cite{debnath2016demonstration,nam2020ground}. The generating sets constructed from our results may also facilitate fault-tolerant quantum computing without the need for swap gates \cite{klaver2024swapless,bombin2016dimensional}. 
%(subsystem) stabilizer codes ~\cite{poulin2005stabilizer,paetznick2013universal} with e.g., transversal Pauli rotations and no need for swap gates.

We demonstrate that for generating sets consisting solely of Pauli strings, the minimal number of operators required to generate $\mathfrak{su}(2^N)$ is $2N+1$. The proof has two components. The first is a general method for constructing universal generating sets on $N$ qubits from sets of operators on a smaller number of qubits, which may be of independent interest. In particular, this method can produce generating sets with $2N+1$ elements. The second establishes that \textit{no} set of $2N$ Pauli strings can generate all of $\mathfrak{su}(2^{N})$. Additionally, we provide an algorithm that, given a generating set as input, outputs a sequence of elements that produce a target Pauli string via nested commutation. This algorithm is demonstrably optimal as it produces sequences of length $\mathcal{O}(N)$, runs in polynomial time and may be of interest for circuit compilation in the contexts mentioned above (i.e. in MBQC and trapped ion quantum computing).
Finally, we analyze the rates at which different sets generate $\mathfrak{su}(2^{N})$, finding that the rate is closely related to the fraction of anticommuting pairs of generators out of the total number of pairs. We give an explicit value for the fraction corresponding to the optimal compilation rate.

In the remainder of the manuscript we provide a brief discourse on relevant aspects of quantum control theory, before presenting our results which include theorems, algorithms, and examples. To conclude, we elaborate on the implications of these results. Further details are given in the Supplemental Material, \cite{supp} (provided here as an appendix). %\footnote{The Supplemental Material contains several proofs of results presented in the main text and includes references \cite{nielsen2010quantum,oeis}.}.

\textit{Background.}---
Let $\mathcal{H}$ denote a finite-dimensional Hilbert space and let $\ket{\psi} \in \mathcal{H}$ be a state. The evolution of $\ket{\psi}$ is governed by the Schrödinger equation 
\begin{align} 
i\hbar \frac{d}{dt} \ket{\psi} = H\ket{\psi}, 
\end{align}
where $H$ is a Hermitian operator. Now suppose we have a set of such operators $\{H_{i}\}$, corresponding to the different ways we can control the quantum system at hand. A natural question to ask is then: Is there a sequence of elements of $\{H_{i}\}$, denoted $H_{l_{1}},...,H_{l_{m}}$, and a sequence of times $t_{l_{1}},...,t_{l_{m}}$ such that evolving $\ket{\psi_{0}}$, first by $H_{l_{1}}$ for time $t_{l_{1}}$ and each subsequent $H_{l_{j}}$ for time $t_{l_{j}}$, ultimately produces $\ket{\psi_{1}}$? Taking each $H_{i}$ to be time-independent, this question is the same as asking whether
\begin{align}
\ket{\psi_{1}} = e^{-iH_{l_{m}}t_{l_{m}}}\dots e^{-iH_{l_{1}}t_{l_{1}}}\ket{\psi_{0}} \label{eq:multiple_controls}
\end{align}
holds for some $H_{l_{1}},...,H_{l_{m}} \in \{H_{i}\}$ and $t_{l_{1}},...,t_{l_{m}} \in \mathbb{R}$ (note we are taking $\hbar = 1$).

A useful tool for answering such a question is the Baker-Campbell-Hausdorff (BCH) formula from Lie algebra theory (see e.g., Ref.~\cite{hall2013lie} for further details on Lie algebras). For a Lie algebra $\mathfrak{g}$ and associated adjoint maps $ad_{A} : \mathfrak{g} \rightarrow \mathfrak{g}$ defined for each $A \in \mathfrak{g}$ via $B \mapsto [A,B]$, the BCH formula can be written as: for $A,B \in \mathfrak{g}$,
\begin{align}
e^{A}e^{B} = e^{A + B + \sum_{s = 1}^{\infty} \frac{1}{s!}ad_{A}^{(s)}(B)}, 
\end{align}
with $ad_{A}^{(s)}(B)$ denoting $ad_{A}$ composed with itself $s$ times and applied to $B$. Taking $A = -iH_{l_{2}}t_{l_{2}}$ and $B = -iH_{l_{1}}t_{l_{1}}$ we see that we can rewrite the right-most terms of the product in (\ref{eq:multiple_controls}) as a real linear combination of $-iH_{l_{1}}$, $-iH_{l_{2}}$, their commutator, the commutator of $-iH_{l_{2}}$ with their commutator, and so on. By repeated use of the BCH formula, the entire product of (\ref{eq:multiple_controls}) can be reduced to a single term whose exponent is a linear combination of the $H_{l_{j}}$ and nested commutators thereof.

So, if our aim is to determine whether there exists a sequence such that (\ref{eq:multiple_controls}) holds for \textit{any} pair of $N$-qubit states $\ket{\psi_{0}}, \ket{\psi_{1}}$ then it suffices to show that the $iH_{l_{j}}$ and their nested commutators generate $\mathfrak{su}(2^{N})$~\cite{ramakrishna1996relation}. Taking $\mathcal{A} := \{i H_{j}\}_{j}$ to be a set of generators, the \textit{dynamical Lie algebra} generated by $\mathcal{A}$ is 
\begin{align}
\langle \mathcal{A} \rangle_{\mathrm{Lie}}:=\mathrm{span}_{\mathbb{R}} \left\{\mathcal{A} \bigcup_{r=1}^{\infty} \mathcal{A}_{ad^{(r)}} \right\}, 
\end{align}
where 
\begin{align}
\mathcal{A}_{ad^{(r)}} := \{ ad_{A_{i_{1}}}\cdots ad_{A_{i_{r}}}(A) | (A_{i_{1}},...,A_{i_{r}},A) \in \mathcal{A}^{r+1}\}. 
\end{align}
The question of controllability between any two states for a given set $\mathcal{A}$ is thus a question of whether $\langle \mathcal{A} \rangle_{\mathrm{Lie}} = \mathfrak{su}(2^{N})$.

%Like any vector space, $\mathfrak{su}(2^{N})$ has a basis. For example, for $N = 1$, we have that $\mathfrak{su}(2)=\textrm{span}_{\mathbb{R}}\{iX, iZ, iY\}$ where $X,Y,Z$ are the Hermitian $2\times 2$ Pauli matrices. Extending to any $N$, let us consider the set of all Pauli strings on $N$ qubits, $\mathcal{P}_N:=\{I, X,Y,Z\}^{\otimes N}$, such that for every element $P\in \mathcal{P}_N$, $P=P_1\otimes P_2\otimes \cdots \otimes P_N$, where $P_i\in\{I,X,Y,Z\}$ acts on qubit $i=1,...,N$ and $I$ is the identity matrix. This set forms a basis (over $\mathbb{R}$) of the vector space of $2^N\times 2^N$ Hermitian matrices. Defining $\mathcal{P}_N^*:=\mathcal{P}_N-\{I^{\otimes N}\}$, we see that $i\mathcal{P}_N^*$ forms a basis of the $2^{N} \times 2^{N}$ traceless, skew-Hermitian matrices and hence $\mathfrak{su}(2^{N}) = \textrm{span}_{\mathbb{R}}(i\mathcal{P}_N^*)$. Defining $\langle \mathcal{A} \rangle_{[\cdot,\cdot]} := \mathcal{A} \bigcup_{r=1}^{\infty} \mathcal{A}_{ad^{(r)}}$, it follows that a sufficient condition for $\langle \mathcal{A} \rangle_{\mathrm{Lie}} = \mathfrak{su}(2^{N})$ to hold for a given set $\mathcal{A}$, is that for each $A \in i\mathcal{P}_{N}^{*}$ there is some non-zero $\alpha \in \mathbb{R}$ such that $\alpha A \in \langle \mathcal{A} \rangle_{[\cdot,\cdot]}$. In such a case, we say that $\mathcal{A}$ is \textit{adjoint universal} for $i\mathcal{P}_{N}^{*}$.

Like any vector space, $\mathfrak{su}(2^{N})$ has a basis. Of particular relevance for this work is the basis derived from the set $\mathcal{P}_{N}$ of all Pauli strings on $N$ qubits. %$:=\{I, X,Y,Z\}^{\otimes N}$, such that for every element $P\in \mathcal{P}_N$, $P=P_1\otimes P_2\otimes \cdots \otimes P_N$, where $P_i\in\{I,X,Y,Z\}$ acts on qubit $i=1,...,N$ and $I$ is the identity matrix. 
This set forms a basis (over $\mathbb{R}$) of the vector space of $2^N\times 2^N$ Hermitian matrices. Defining $\mathcal{P}_N^*:=\mathcal{P}_N-\{I^{\otimes N}\}$, we have that $i\mathcal{P}_N^*$ forms a basis of the $2^{N} \times 2^{N}$ traceless, skew-Hermitian matrices and hence $\mathfrak{su}(2^{N}) = \textrm{span}_{\mathbb{R}}(i\mathcal{P}_N^*)$. Defining 
\begin{align}
\langle \mathcal{A} \rangle_{[\cdot,\cdot]} := \mathcal{A} \bigcup_{r=1}^{\infty} \mathcal{A}_{ad^{(r)}}, 
\end{align}
it follows that a sufficient condition for $\langle \mathcal{A} \rangle_{\mathrm{Lie}} = \mathfrak{su}(2^{N})$ to hold for a given set $\mathcal{A}$, is that for each $A \in i\mathcal{P}_{N}^{*}$ there is some $\alpha \in \mathbb{R}^{*}$ such that $\alpha A \in \langle \mathcal{A} \rangle_{[\cdot,\cdot]}$. In such a case, we say that $\mathcal{A}$ is \textit{adjoint universal} for $i\mathcal{P}_{N}^{*}$.

Here, we consider the question: how small can $\mathcal{A}\subset i\mathcal{P}_N^*$ be such that
$\langle \mathcal{A} \rangle_{\mathrm{Lie}}=\mathfrak{su}(2^N)$? As noted in the introduction, if we do not enforce that $\mathcal{A}\subset i\mathcal{P}_N^*$ but allow $\mathcal{A}$ to be any set of traceless, skew-Hermitian matrices, then the answer to the analogous question is known to be $2$ \cite{Kuranishi1951On}. 

The restriction to considering subsets of $i\mathcal{P}_N^*$ comes with certain convenient features worth highlighting. First, due to the commutation relations of Pauli strings, the matrix commutator of elements of $i\mathcal{P}_{N}^{*}$ is again an element of $i\mathcal{P}_{N}^{*}$ up to some real scalar. Pursuant to the discussion above, when working with generating sets $\mathcal{A} \subset i\mathcal{P}_{N}^{*}$, $\langle \mathcal{A} \rangle_{[\cdot, \cdot]} = i\mathcal{P}_{N}^{*}$ is in fact also a necessary condition for $\langle \mathcal{A} \rangle_{\textrm{Lie}} = \mathfrak{su}(2^{N})$. This allows us to work at the level of $\langle \mathcal{A} \rangle_{[\cdot, \cdot]}$ rather than the full DLA $\langle \mathcal{A} \rangle_{\textrm{Lie}}$.

Second, as $i\mathcal{P}_{N}^{*}$ is related to the Pauli group on $N$ qubits, there is machinery developed for the latter that can be leveraged to make statements about the former. Specifically, there is a canonical mapping between the $N$-qubit Pauli group and the additive group $\mathbb{F}_{2}^{2N}$, which we make use of to prove certain results below. This mapping is commonly used in stabiliser quantum theory (see e.g., \cite{Nest_04,Aaronson_04}); the details are provided in \cite{supp}.

\textit{Minimal generating sets of $\mathfrak{su}(2^N)$.}---
We now demonstrate that for all $\mathcal{A}\subset i\mathcal{P}_N^*$ such that $\langle \mathcal{A} \rangle_{\mathrm{Lie}}=\mathfrak{su}(2^N)$, $|\mathcal{A}| \geq 2N+1$, which is tight. Moreover, we provide a construction for minimal such sets and an algorithm to optimally generate any other Pauli string from them. To this end, we require a notion similar to adjoint universality discussed earlier: a set $\mathcal{B} \subset i\mathcal{P}_{N}^{*}$ is \textit{product universal} if the set 
\begin{align}
\langle \mathcal{B} \rangle_{\times} := \bigcup_{r = 1}^{\infty} \{B_{i_{1}}\cdot B_{i_{2}}\cdot ... \cdot B_{i_{r}} | (B_{i_{1}},...,B_{i_{r}}) \in \mathcal{B}^{r}\}
\end{align}
is such that for each $A \in i\mathcal{P}_{N}^{*}$ there is an $\alpha \in \{\pm 1, \pm i\}$ such that $\alpha A \in \langle \mathcal{B} \rangle_{\times}$ \footnote{Note the difference in the allowed set of scalars compared to that for adjoint universality}. Our first result is:
\begin{Theorem}\label{thm:construction}
    Fix $k \geq 2$ and consider $\mathcal{A} \subset i\mathcal{P}_k^*$  and $\mathcal{B}\subset i\mathcal{P}^*_{N-k}$ such that $\langle\mathcal{A}\rangle_{\mathrm{Lie}}=\mathfrak{su}(2^k)$ and $\mathcal{B}$ is product universal for $i\mathcal{P}^*_{N-k}$. Let $\mathbb{B} := \{(U_{B}, B)\}_{B \in \mathcal{B}}$ be a set of pairs  where each $U_{B} \in \mathcal{P}_{k}^{*}$. Defining $\mathcal{A}'=\{A\otimes I^{\otimes N-k}|A\in\mathcal{A}\}$ and $\mathcal{B}'=\{U_{B} \otimes B | (U_{B}, B) \in \mathbb{B}\}$, we have that $\langle \mathcal{A}'\cup\mathcal{B}'\rangle_{\mathrm{Lie}}=\mathfrak{su}(2^N)$.
\end{Theorem}
The proof of this theorem is given in \cite{supp} and relies on the correctness of the algorithm \textsc{PauliCompiler} (\Cref{alg:paulicompiler}). Simply put, the above theorem allows universal generating sets on $N$ qubits to be constructed from a universal generating set on a constant-size subset of qubits and a product universal set on the remaining qubits.
%to be constructed from a universal generating set on a smaller number of qubits. 
To arrive at our claim that there exist universal generating subsets of $i\mathcal{P}_{N}^{*}$ consisting of $2N+1$ elements, we consider the following:
\begin{Proposition} \label{prop:su_4} For $\mathcal{A} = \{iX_{1}, iZ_{1}, iX_{2}, iZ_{2}, iZ_{1} \otimes Z_{2}\}$, we have that $\langle \mathcal{A} \rangle_{\textrm{Lie}} = \mathfrak{su}(4)$.
\end{Proposition}
The proof can be established by direct calculation that $\langle \mathcal{A} \rangle_{[\cdot, \cdot]} = i\mathcal{P}_{2}^{*}$ (c.f. \cite{supp}). Noting that there exist product universal sets for $i\mathcal{P}_{m}^{*}$ consisting of $2m$ elements, we thus have the following:
\begin{Corollary} \label{cor:2N+1} For $\mathcal{A}$ as in \Cref{prop:su_4}, and $\mathcal{A}'$ and $\mathcal{B}'$ as in \Cref{thm:construction} with $|\mathcal{B}'|=(2N-2)$, $|\mathcal{A}' \cup \mathcal{B}'| = 2N+1$.
\end{Corollary}

Having shown that we can generate $\mathfrak{su}(2^{N})$ from a set $\mathcal{A} \subset i\mathcal{P}_{N}^{*}$ such that $|\mathcal{A}| = 2N+1$, we now show that no strictly smaller set suffices. In fact, we need only check that no generating set $\mathcal{A} \subset i \mathcal{P}_{N}^{*}$ of $2N$ elements is able to produce the entire Lie algebra, since if $\mathcal{A}$ is insufficient, any subset of $\mathcal{A}$ is as well.
%for the following reasons. For any $\mathcal{A}, \mathcal{A}' \subset i\mathcal{P}_{N}^{*}$ such that $\mathcal{A} \subseteq \mathcal{A}'$, we have that $\langle \mathcal{A} \rangle_{[\cdot, \cdot]} \subseteq \langle \mathcal{A}' \rangle_{[\cdot, \cdot]}$ and so a proof that any set of size $2N$ is insufficient immediately implies that any set contained within it is insufficient also. 
The following theorem shows that it is indeed impossible to generate $\mathfrak{su}(2^{N})$ using $2N$ elements of $i\mathcal{P}_{N}^{*}$: 
\begin{Theorem}\label{thm:2n} Let $\mathcal{A}\subset i\mathcal{P}_{N}^{*}$ consist of $2N$ elements. Then, $\langle \mathcal{A}\rangle_{\mathrm{Lie}}\neq \mathfrak{su}(2^N)$.
\end{Theorem}
The proof is given in \cite{supp} and consists in showing that there exists an element $A$ of $i\mathcal{P}_{N}^{*}$ such that $\alpha A \notin \langle \mathcal{A}\rangle_{[\cdot, \cdot]}$ for any non-zero $\alpha \in \mathbb{R}$. The proof proceeds by considering two distinct cases for $\mathcal{A}$, namely where all elements mutually anticommute and where there exists at least one pair that commute. In both cases, an element of $i\mathcal{P}_{*}^{N}$ not contained in $\langle \mathcal{A} \rangle_{[\cdot, \cdot]}$ is shown to exist. This proof makes use of the aforementioned connection between the Pauli group on $N$ qubits and $\mathbb{F}_{2}^{2N}$ (cf. \cite{supp}).

We now turn to several examples related to the results presented so far. 

\textbf{Example 1:} Consider the following sets: 
\begin{gather*}
    \mathcal{A}=\{iX_1,iZ_1,iX_2,iZ_2, iZ_1\otimes Z_2\},\nonumber\\
    \mathcal{B}'_1=\{iX_2\otimes\prod_{2<j<i}Y_j\otimes Z_i, iZ_2\otimes\prod_{2<j<i}Y_j\otimes X_i\}_{i=3,...,N}.
\end{gather*} 
These sets fulfill the conditions of Theorem~\ref{thm:construction}, and since $|\mathcal{B}'_{1}| =  2(N-2)$, their union gives an adjoint universal set of size $2N+1$. This set plays an interesting role in the next example.  

\textbf{Example 2:} Take $\mathcal{A}$ to be as in the previous example, but instead of $\mathcal{B}'_{1}$, consider
\begin{align}
\mathcal{B}'_2=\{iX_i\otimes Z_{i+1}, iZ_i\otimes X_{i+1}\}_{i=2,...,N-1}. 
\end{align}
The union $\mathcal{A} \cup \mathcal{B}'_{2}$ is particularly interesting since it contains only nearest neighbor interactions and has $2N+1$ elements. However, the conditions of Theorem~\ref{thm:construction} are not fulfilled in this case as not all elements of $\mathcal{B}'_2$ have support on qubits 1 and 2. Fortunately, we can show $\mathcal{B}'_1\subset \langle \mathcal{B}_2'\rangle_{[\cdot,\cdot]}$ (up to some factors), as follows:
\begin{align}\label{eq:example2_mapping}
    iX_2\otimes\prod_{2<j<i}Y_j\otimes Z_i&\propto ad_{X_2Z_3}\cdots ad_{X_{i-2}Z_{i-1}}(X_{i-1}Z_i) \\
    iZ_2\otimes\prod_{2<j<i}Y_j\otimes X_i&\propto ad_{Z_2X_3}\cdots ad_{Z_{i-2}X_{i-1}}(Z_{i-1}X_i)
\end{align}
Since  $\mathcal{A}'\cup \mathcal{B}_1'$ is adjoint universal, so is $\mathcal{A}'\cup \mathcal{B}_2'$. This demonstrates that,  Theorem~\ref{thm:construction} not only provides us with a method for constructing universal generating set, it also allows us to prove universality for other \textit{a priori} distinct generating sets. 

\textbf{Example 3:} Our third example of a universal set of minimal size $2N+1$ has previously emerged as the native gateset of a particular scheme of MBQC \cite{Stephen2024}. 
Here we summarize the distinct properties of this set, with a full description given in \cite{supp}. First, the elements of the set essentially all have the form $T^{\ell\dagger} Z_1 T^{\ell}$ where $\ell\in \mathbb{N}$ and $T$ is a finite-depth circuit consisting of nearest-neighbour Clifford gates. Therefore, rather than requiring native implementation of the rotations generated by many different Pauli strings, this gateset only requires implementation of the rotation $e^{i\theta Z_1}$ and the Clifford circuit $T$. 
Interestingly, we will see that this set generates all Pauli strings at faster rate than the other examples.

\begin{figure}[t!]
\begin{algorithm}[H]
\caption{\textsc{PauliCompiler}$(P)$}
\begin{algorithmic}
\Require Target: $P=iV\otimes W\in i\mathcal{P}^*_{N}$ for $V\in\mathcal{P}_k,W\in\mathcal{P}_{N-k}$.\\ \hspace{1.2cm} $\mathcal{A},\mathcal{A}',\mathcal{B},\mathcal{B}'$ as in Theorem~\ref{thm:construction}. \vspace{0.2cm}
\If{$W=I$}
\State \begin{varwidth}[t]{\linewidth}
Choose $A_1,...,A_s\in \mathcal{A}$ s.t. $iV\propto ad_{A_1}\cdots ad_{A_{s-1}}(A_s)$\par \hspace{1.8cm}$\wedge\:s=\mathcal{O}(1)$
\end{varwidth}\\
\Return $(A_1\otimes I,...,A_s\otimes I)$
\ElsIf{$V=I$}
\State Choose $W_1, W_2\in i\mathcal{P}_{N-k}$ s.t. $iW\propto [W_1,W_2]$
\State \begin{varwidth}[t]{\linewidth}
$\mathcal{G}'=(G_1',...,G_{\mathcal{G}'}')\gets\textsc{SubsystemCompiler}(W_1)$\par \hspace{0.7cm} s.t. $V_1'\otimes W_1\propto ad_{G'_{|\mathcal{G}'|}}\cdots ad_{G_2 '}(G_1')$
\end{varwidth}
\State \begin{varwidth}[t]{\linewidth}
$\mathcal{G}''=(G_1'',...,G_{\mathcal{G}''}'')\gets\textsc{SubsystemCompiler}(W_2)$\par \hspace{0.7cm} s.t. $V_2'\otimes W_2\propto ad_{G''_{|\mathcal{G}''|}}\cdots ad_{G_2''}(G_1'')$
\end{varwidth}
\State \begin{varwidth}[t]{\linewidth}
Choose $A_1,...,A_s\in \mathcal{A}$ s.t. $iV_1'\propto ad_{A_1}\cdots ad_{A_{s}}(iV_2')$\par \hspace{2.05cm}$\wedge\:s=\mathcal{O}(1)$
\end{varwidth}
\State $\mathcal{G}=(G_1,...,G_{|\mathcal{G}|})\equiv\mathcal{G}'\oplus \mathcal{G}''\oplus (A_1\otimes I,...,A_s\otimes I)$
\State Choose $\sigma\in S_{|\mathcal{G}|}$ s.t. $ad_{G_{\sigma(1)}}\cdots ad_{G_{\sigma(|\mathcal{G}|-1)}}(G_{\sigma(|\mathcal{G}|)})\neq 0$\\
\Return $(G_{\sigma(1)},...,G_{\sigma(|\mathcal{G}|)})$
\Else
\State \begin{varwidth}[t]{\linewidth}
$\mathcal{G}'=(G_1',...,G_{\mathcal{G}'}')\gets\textsc{SubsystemCompiler}(iW)$\par \hspace{0.7cm} s.t. $iV'\otimes W\propto ad_{G'_{|\mathcal{G}'|}}\cdots ad_{G_2 '}(G_1')$
\end{varwidth}
\State \begin{varwidth}[t]{\linewidth}
Choose $A_1,...,A_s\in \mathcal{A}$ s.t. $iV\propto ad_{A_1}\cdots ad_{A_{s}}(iV')$\par \hspace{2.15cm}$\wedge\:s=\mathcal{O}(1)$
\end{varwidth}\\
\Return $({A_1\otimes I,...,A_s\otimes I,G'_{|\mathcal{G}'|},...,G'_1})$
\EndIf
\end{algorithmic}
\label{alg:paulicompiler}
\end{algorithm}
\end{figure}

\textit{Optimal compiler for \Cref{thm:construction}.}---Given an adjoint universal set for $i\mathcal{P}_{N}^{*}$, one can ask \textit{how} a specific Pauli string can be obtained via a sequence of adjoint maps from elements of the set. Here, we turn to the question of compilation of Pauli strings for generating sets of the form considered in \Cref{thm:construction}. First, we need the following result, proven in \cite{supp}:
\begin{Proposition}\label{prop:Vmap}
    Let $V,V'\in i\mathcal{P}^*_k$ and $\mathcal{A}$ be an adjoint universal set for $i\mathcal{P}^*_k$ for any $k\geq 1$. Then, there exist $A_1,...,A_r\in\mathcal{A}$ and $\alpha\in\mathbb{R}^*$ such that $V=\alpha\: ad_{A_1}\cdots ad_{A_{r}}(V')$.
\end{Proposition}
This proposition is a statement about the existence of a sequence to produce an element $V$ from a fixed $V'$. In what follows, we will be interested in the length of such sequences when only the target element is specified, i.e. when $V'$ is allowed to vary. In particular, we consider sequences from generating sets of the form $\mathcal{A}' \cup \mathcal{B}'$ with $\mathcal{A}'$ and $\mathcal{B}'$ as in \Cref{thm:construction}. 

Naively, a breadth-first search will give sequences of optimal length for any target $V$. However this process will take exponentially long in the number of qubits in the worst case. In contrast, \textsc{PauliCompiler} (\Cref{alg:paulicompiler}) takes quadratic time in the number of qubits for any desired Pauli string (as shown in \cite{supp}):
\begin{Theorem}\label{thm:optimal}
     Consider a set of operators $\mathcal{A}'\cup \mathcal{B}'$ as in Theorem~\ref{thm:construction} such that $|\mathcal{A}'\cup \mathcal{B}'|=2N+1$. Let $P\in i\mathcal{P}^*_N$. Then, \textsc{PauliCompiler}($P$) returns a sequence of operators $G_1,...,G_L\in \mathcal{A}'\cup \mathcal{B}'$ such that $P=\alpha\: ad_{G_1}\cdots ad_{G_{L-1}}(G_L)$ for some $\alpha\in\mathbb{R}^*$ and $L=\mathcal{O}(N)$.
\end{Theorem}
The proof is given in \cite{supp}. That $\mathcal{O}(N)$ is indeed optimal for producing an arbitrary Pauli string via nested commutators follows from the fact that there exist Pauli strings that can only be produced by no less than a product of $N$ elements from a product universal set.

A central component of \textsc{PauliCompiler} is the algorithm \textsc{SubsystemCompiler}. This algorithm produces a sequence of elements that generates the desired operator $W\in\mathcal{P}_{N-k}$ on one subsystem, but with a potentially undesired operator $V'\in\mathcal{P}^*_k$ on the other subsystem (cf. \cite{supp} for details). In effect, this does not pose a problem due to the assumptions on $\mathcal{A}$ being adjoint universal and \Cref{prop:Vmap}.

\begin{figure}
    \centering
    \includegraphics[width=\linewidth]{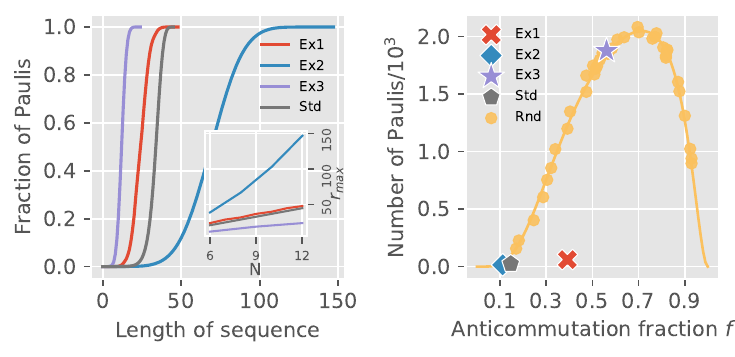}
    \caption{Left: Comparison of the rates at which different sets of Pauli strings generate $\mathfrak{su}(2^N)$ for $N=12$. Ex 1, 2, 3 refer to the three examples of the main text, while Std refers to the standard gate set. Inset: The maximum length $r_{\text{max}}$ of commutator sequences needed to generate all Pauli strings as a function of number of qubits. Right: Number of Paulis generated after five rounds of nested commutators on $N=9$ qubits versus the fraction of pairs of generators that anticommute. Each round point describes a random adjoint universal set of Paulis with minimal length $2N+1$. The solid line is a theoretical prediction described in \cite{supp}.}
    \label{fig:rates}
\end{figure}

\textit{Compilation efficiency of generating sets.}---
While \textsc{PauliCompiler} demonstrates the optimal $\mathcal{O}(N)$ scaling of the maximum value of $L$ in Theorem \ref{thm:optimal}, it says nothing regarding the possible pre-factors, i.e. compilation efficiency, for different generating sets.
 %representing the rate at which $\mathfrak{su}(2^N)$ is generated.
  %leaves room for constants that would allow certain generating sets to generate $\mathfrak{su}(2^N)$ at a faster rate. 
Figure ~\ref{fig:rates} depicts the Pauli string generation rates for the three examples given above, alongside the ``standard gate set'' consisting of single-qubit gates and nearest-neighbor two-qubit gates on a one-dimensional array which corresponds to taking 
\begin{align}
\mathcal{A} = \{ X_i,Z_i\}_{i=1,\dots, N}\cup \{Z_i\otimes Z_{i+1}\}_{i=1,\dots,N-1}. 
\end{align}
This set is adjoint universal, but consists of $3N-1$ generators and therefore is not minimal. 

We make three observations regarding the results of Fig.~\ref{fig:rates}. First, each case (except Example 2) generates $\mathfrak{su}(2^N)$ after $\mathcal{O}(N)$ steps as expected, but the pre-factors do differ. In particular, Example 3 requires $\sim 2N$ commutators to generate all Paulis while the standard gate set requires $\sim 4N$, despite the fact that the former set has fewer generators to choose from. Second, the maximum length $L$ in Example 2 scales as $\mathcal{O}(N^2)$. This is consistent with the observation that Example 2 must first be mapped onto Example 1 using $\mathcal{O}(N)$ generators in Eq.~\ref{eq:example2_mapping} to use \textsc{PauliCompiler}, suggesting that this is an optimal recipe for compilation.

Finally, we find that the Pauli string generation rate is closely linked to the fraction $f$ of anticommuting pairs of generators. In \cite{supp}, we calculate the generation rate as a function of this fraction and find that the optimal fraction $f\approx 0.706$ can be determined analytically as the value that maximizes a particular $q$-Pochhammer symbol. This calculation is in strong agreement with the generation rate of random adjoint universal sets of Pauli strings (see Fig.~\ref{fig:rates}) and it explains why Example 3, which has $f\approx 0.5$, generates Pauli strings faster than Example 1 and the standard gate set, which each have $f\rightarrow 0$ for large $N$~\footnote{Example 1 does not follow the curve since it has a non-uniform distribution of anticommutation; some generators anticommute with half the others while some only anticommute with a constant number of others.}.
It is an interesting open question to explicitly construct sets of Pauli strings that generate all of $\mathfrak{su}(2^N)$ at the optimal rate.

\textit{Discussion.}---
To conclude, let us discuss several implications of the above results. First note that, due to the commutation properties of Pauli strings, the BCH formula takes on a neat form: for $A,B \in i\mathcal{P}_{N}^{*}$ such that $[A,B] \neq 0$, we have (see \cite{supp})
\begin{align}
e^{\frac{\pi}{4}A}e^{t B}e^{-\frac{\pi}{4}A} = e^{\frac{t}{2} [A,B] }. \label{eq:Pauli_BCH_trick}
\end{align}
By iteratively applying this formula, we observe that, for a universal generating set $\mathcal{A} \subset i\mathcal{P}_{N}^{*}$ and any $P \in \mathcal{P}_{N}^{*}$, there exists a sequence $G_{1}, ..., G_{L} \in \mathcal{A}$ such that
\begin{align}
e^{i\alpha P} = e^{\frac{\pi}{4}G_1}\cdots e^{\frac{\pi}{4}G_{L-1}} e^{t G_L}e^{-\frac{\pi}{4}G_{L-1}}\cdots e^{-\frac{\pi}{4}G_{1}},\label{eq:ad_gate}
\end{align}
where $t = 2^{L-1}\alpha$. From the results above, we know how to produce the sequence $G_{1}, ..., G_{L}$ and furthermore that $L = \mathcal{O}(N)$. 

This observation is likely of interest in areas related to MBQC and fault-tolerant quantum computing for the following reasons. For the former, it is known that any Clifford operation, such as the $e^{\pm i \frac{\pi}{4}G_{j}}$, can be performed in a single round of measurements. Accordingly, implementations of Pauli strings via Eq.~(\ref{eq:ad_gate}) for the optimal sequences produced by the \textsc{PauliCompiler} may lead to optimal resource states in MBQC (where long Pauli strings are common - see e.g., \cite{Raussendorf2019,nautrup2023,Stephen2024,baeumer2024measurement}).

For the latter, our results may lead to the design of new (subsystem) quantum error correcting codes~\cite{poulin2005stabilizer}. For example, by taking $e^{tG_{L}}$ to be a fixed non-Clifford gate that can be implemented transversally~\cite{paetznick2013universal}, Eq.~(\ref{eq:ad_gate}) and \Cref{prop:Vmap} allow us to conclude that {\em any} non-Clifford Pauli rotation can be implemented using only subsequent Clifford operations, which can themselves be implemented transversally in e.g., the color code \cite{bombin2006topological,beverland2021cost}. Then, our observations on generation rate may aid the design of an optimal generating set of Cliffords while \textsc{PauliCompiler} allows us to optimize the number of required Clifford gates per non-Clifford gate.
From this perspective, the set from Example 3 is particularly interesting due to it's structure (requiring a single non-Clifford gate only on the first qubit, $G_L\equiv Z_1$) and it's efficiency (recall Fig.~\ref{fig:rates}). Notably, this method foregoes the explicit use of SWAP-gates, which can improve overall efficiency~\cite{klaver2024swapless}, and may be directly applied to existing architectures such as stacks of color codes \cite{bombin2016dimensional}

Rotations of Pauli strings also arise in a number of other areas of quantum computation, to which our results may also be applicable. For example, various methods for simulating fermionic systems on a quantum computer \cite{steudtner2018fermion,romero2018strategies,cowtan2020generic,van2020circuit} or more generally performing the time evolution of a given Hamiltonian, often result in circuits consisting of many Pauli strings after application of the Lie-Trotter formula. Similarly, in the context of quantum optimisation, problem constraints are usually mapped to Pauli strings of weight two or higher \cite{lucas2014ising,farhi2014quantumapproximateoptimizationalgorithm}, again leading to circuits containing Pauli string rotations. Finally, in trapped ions systems it is common to have a gate set comprising single-qubit gates, two-qubit $X\otimes X$-rotations \cite{debnath2016demonstration}, and possibly $N$-body Pauli rotations \cite{Katz2022}, which have featured in e.g., variational approaches to finding ground states of molecules \cite{nam2020ground}. A better understanding of the universality requirements and compilation rates of gate sets comprised of Pauli strings will likely lead to a better utilization of resources in each of these cases.

\begin{acknowledgments}
The authors thank Vedran Dunjko, Paul Barth and Arunava Majumder for stimulating discussions. This project was funded in whole or in part by the Austrian Science Fund (FWF) [DK-ALM W$1259$-N$27$ (Grant DOI: 10.55776/W1259), SFB BeyondC F7102 (Grant DOI: 10.55776/F71), WIT9503323 (Grant DOI: 10.55776/WIT9503323)]. For open access purposes, the authors have applied a CC BY public copyright license to any author-accepted manuscript version arising from this submission. This work was also co-funded by the European Union (ERC, QuantAI, Project No. $101055129$, Grant DOI: 10.3030/101055129). Views and opinions expressed are however those of the author(s) only and do not necessarily reflect those of the European Union or the European Research Council. Neither the European Union nor the granting authority can be held responsible for them. DTS is supported by the Simons Collaboration on Ultra-Quantum Matter, which is a grant from the Simons Foundation (651440).

During the preparation of this manuscript, we became aware of a related work, Ref.~\cite{aguilar2024classificationpauliliealgebras}, which contains some similar results to our own. However, both works also feature substantially different results and methods. It would be an interesting avenue for future work to determine whether the different approaches could be combined to make further progress in this area.
\end{acknowledgments}

\bibliography{bib}

\begin{thebibliography}{49}%
\makeatletter
\providecommand \@ifxundefined [1]{%
 \@ifx{#1\undefined}
}%
\providecommand \@ifnum [1]{%
 \ifnum #1\expandafter \@firstoftwo
 \else \expandafter \@secondoftwo
 \fi
}%
\providecommand \@ifx [1]{%
 \ifx #1\expandafter \@firstoftwo
 \else \expandafter \@secondoftwo
 \fi
}%
\providecommand \natexlab [1]{#1}%
\providecommand \enquote  [1]{``#1''}%
\providecommand \bibnamefont  [1]{#1}%
\providecommand \bibfnamefont [1]{#1}%
\providecommand \citenamefont [1]{#1}%
\providecommand \href@noop [0]{\@secondoftwo}%
\providecommand \href [0]{\begingroup \@sanitize@url \@href}%
\providecommand \@href[1]{\@@startlink{#1}\@@href}%
\providecommand \@@href[1]{\endgroup#1\@@endlink}%
\providecommand \@sanitize@url [0]{\catcode `\\12\catcode `\$12\catcode `\&12\catcode `\#12\catcode `\^12\catcode `\_12\catcode `\%12\relax}%
\providecommand \@@startlink[1]{}%
\providecommand \@@endlink[0]{}%
\providecommand \url  [0]{\begingroup\@sanitize@url \@url }%
\providecommand \@url [1]{\endgroup\@href {#1}{\urlprefix }}%
\providecommand \urlprefix  [0]{URL }%
\providecommand \Eprint [0]{\href }%
\providecommand \doibase [0]{https://doi.org/}%
\providecommand \selectlanguage [0]{\@gobble}%
\providecommand \bibinfo  [0]{\@secondoftwo}%
\providecommand \bibfield  [0]{\@secondoftwo}%
\providecommand \translation [1]{[#1]}%
\providecommand \BibitemOpen [0]{}%
\providecommand \bibitemStop [0]{}%
\providecommand \bibitemNoStop [0]{.\EOS\space}%
\providecommand \EOS [0]{\spacefactor3000\relax}%
\providecommand \BibitemShut  [1]{\csname bibitem#1\endcsname}%
\let\auto@bib@innerbib\@empty
%</preamble>
\bibitem [{\citenamefont {Huang}\ \emph {et~al.}(1983)\citenamefont {Huang}, \citenamefont {Tarn},\ and\ \citenamefont {Clark}}]{huang1983controllability}%
  \BibitemOpen
  \bibfield  {author} {\bibinfo {author} {\bibfnamefont {G.~M.}\ \bibnamefont {Huang}}, \bibinfo {author} {\bibfnamefont {T.~J.}\ \bibnamefont {Tarn}},\ and\ \bibinfo {author} {\bibfnamefont {J.~W.}\ \bibnamefont {Clark}},\ }\bibfield  {title} {\bibinfo {title} {{On the controllability of quantum-mechanical systems}},\ }\href@noop {} {\bibfield  {journal} {\bibinfo  {journal} {Journal of Mathematical Physics}\ }\textbf {\bibinfo {volume} {24}},\ \bibinfo {pages} {2608} (\bibinfo {year} {1983})}\BibitemShut {NoStop}%
\bibitem [{\citenamefont {Jurdjevic}\ and\ \citenamefont {Sussmann}(1972)}]{jurdjevic1972control}%
  \BibitemOpen
  \bibfield  {author} {\bibinfo {author} {\bibfnamefont {V.}~\bibnamefont {Jurdjevic}}\ and\ \bibinfo {author} {\bibfnamefont {H.~J.}\ \bibnamefont {Sussmann}},\ }\bibfield  {title} {\bibinfo {title} {{Control systems on Lie groups}},\ }\href@noop {} {\bibfield  {journal} {\bibinfo  {journal} {Journal of Differential equations}\ }\textbf {\bibinfo {volume} {12}},\ \bibinfo {pages} {313} (\bibinfo {year} {1972})}\BibitemShut {NoStop}%
\bibitem [{\citenamefont {Kuranishi}(1951)}]{Kuranishi1951On}%
  \BibitemOpen
  \bibfield  {author} {\bibinfo {author} {\bibfnamefont {M.}~\bibnamefont {Kuranishi}},\ }\bibfield  {title} {\bibinfo {title} {{On Everywhere Dense Imbedding of Free Groups in Lie Groups}},\ }\href {https://api.semanticscholar.org/CorpusID:117066619} {\bibfield  {journal} {\bibinfo  {journal} {Nagoya Mathematical Journal}\ }\textbf {\bibinfo {volume} {2}},\ \bibinfo {pages} {63} (\bibinfo {year} {1951})}\BibitemShut {NoStop}%
\bibitem [{\citenamefont {Ramakrishna}\ \emph {et~al.}(1995)\citenamefont {Ramakrishna}, \citenamefont {Salapaka}, \citenamefont {Dahleh}, \citenamefont {Rabitz},\ and\ \citenamefont {Peirce}}]{Ramakrishna}%
  \BibitemOpen
  \bibfield  {author} {\bibinfo {author} {\bibfnamefont {V.}~\bibnamefont {Ramakrishna}}, \bibinfo {author} {\bibfnamefont {M.~V.}\ \bibnamefont {Salapaka}}, \bibinfo {author} {\bibfnamefont {M.}~\bibnamefont {Dahleh}}, \bibinfo {author} {\bibfnamefont {H.}~\bibnamefont {Rabitz}},\ and\ \bibinfo {author} {\bibfnamefont {A.}~\bibnamefont {Peirce}},\ }\bibfield  {title} {\bibinfo {title} {{Controllability of molecular systems}},\ }\href {https://doi.org/10.1103/PhysRevA.51.960} {\bibfield  {journal} {\bibinfo  {journal} {Phys. Rev. A}\ }\textbf {\bibinfo {volume} {51}},\ \bibinfo {pages} {960} (\bibinfo {year} {1995})}\BibitemShut {NoStop}%
\bibitem [{\citenamefont {Albertini}\ and\ \citenamefont {D'Alessandro}(2001)}]{albertini2001notions}%
  \BibitemOpen
  \bibfield  {author} {\bibinfo {author} {\bibfnamefont {F.}~\bibnamefont {Albertini}}\ and\ \bibinfo {author} {\bibfnamefont {D.}~\bibnamefont {D'Alessandro}},\ }\bibfield  {title} {\bibinfo {title} {{Notions of controllability for quantum mechanical systems}},\ }in\ \href@noop {} {\emph {\bibinfo {booktitle} {Proceedings of the 40th IEEE Conference on Decision and Control (Cat. No. 01CH37228)}}},\ Vol.~\bibinfo {volume} {2}\ (\bibinfo {organization} {IEEE},\ \bibinfo {year} {2001})\ pp.\ \bibinfo {pages} {1589--1594}\BibitemShut {NoStop}%
\bibitem [{\citenamefont {Albertini}\ and\ \citenamefont {D'Alessandro}(2002)}]{albertini2002lie}%
  \BibitemOpen
  \bibfield  {author} {\bibinfo {author} {\bibfnamefont {F.}~\bibnamefont {Albertini}}\ and\ \bibinfo {author} {\bibfnamefont {D.}~\bibnamefont {D'Alessandro}},\ }\bibfield  {title} {\bibinfo {title} {{The Lie algebra structure and controllability of spin systems}},\ }\href@noop {} {\bibfield  {journal} {\bibinfo  {journal} {Linear algebra and its applications}\ }\textbf {\bibinfo {volume} {350}},\ \bibinfo {pages} {213} (\bibinfo {year} {2002})}\BibitemShut {NoStop}%
\bibitem [{\citenamefont {Schirmer}\ \emph {et~al.}(2001)\citenamefont {Schirmer}, \citenamefont {Fu},\ and\ \citenamefont {Solomon}}]{Schirmer_01}%
  \BibitemOpen
  \bibfield  {author} {\bibinfo {author} {\bibfnamefont {S.~G.}\ \bibnamefont {Schirmer}}, \bibinfo {author} {\bibfnamefont {H.}~\bibnamefont {Fu}},\ and\ \bibinfo {author} {\bibfnamefont {A.~I.}\ \bibnamefont {Solomon}},\ }\bibfield  {title} {\bibinfo {title} {{Complete controllability of quantum systems}},\ }\href {https://doi.org/10.1103/PhysRevA.63.063410} {\bibfield  {journal} {\bibinfo  {journal} {Phys. Rev. A}\ }\textbf {\bibinfo {volume} {63}},\ \bibinfo {pages} {063410} (\bibinfo {year} {2001})}\BibitemShut {NoStop}%
\bibitem [{\citenamefont {D'Alessandro}(2007)}]{d2007introduction}%
  \BibitemOpen
  \bibfield  {author} {\bibinfo {author} {\bibfnamefont {D.}~\bibnamefont {D'Alessandro}},\ }\href {https://books.google.sm/books?id=HbMYmAEACAAJ} {\emph {\bibinfo {title} {{Introduction to Quantum Control and Dynamics}}}},\ Chapman \& Hall/CRC Applied Mathematics \& Nonlinear Science\ (\bibinfo  {publisher} {Taylor \& Francis},\ \bibinfo {year} {2007})\BibitemShut {NoStop}%
\bibitem [{\citenamefont {Zeier}\ and\ \citenamefont {Schulte-Herbrüggen}(2011)}]{Zeier2011}%
  \BibitemOpen
  \bibfield  {author} {\bibinfo {author} {\bibfnamefont {R.}~\bibnamefont {Zeier}}\ and\ \bibinfo {author} {\bibfnamefont {T.}~\bibnamefont {Schulte-Herbrüggen}},\ }\bibfield  {title} {\bibinfo {title} {{Symmetry principles in quantum systems theory}},\ }\href {https://doi.org/10.1063/1.3657939} {\bibfield  {journal} {\bibinfo  {journal} {Journal of Mathematical Physics}\ }\textbf {\bibinfo {volume} {52}},\ \bibinfo {pages} {113510} (\bibinfo {year} {2011})},\ \Eprint {https://arxiv.org/abs/https://pubs.aip.org/aip/jmp/article-pdf/doi/10.1063/1.3657939/16123044/113510\_1\_online.pdf} {https://pubs.aip.org/aip/jmp/article-pdf/doi/10.1063/1.3657939/16123044/113510\_1\_online.pdf} \BibitemShut {NoStop}%
\bibitem [{\citenamefont {Sussmann}\ and\ \citenamefont {Jurdjevic}(1972)}]{sussmann1972controllability}%
  \BibitemOpen
  \bibfield  {author} {\bibinfo {author} {\bibfnamefont {H.~J.}\ \bibnamefont {Sussmann}}\ and\ \bibinfo {author} {\bibfnamefont {V.}~\bibnamefont {Jurdjevic}},\ }\bibfield  {title} {\bibinfo {title} {{Controllability of nonlinear systems.}},\ }\href@noop {} {\bibfield  {journal} {\bibinfo  {journal} {Journal of Differential Equations}\ }\textbf {\bibinfo {volume} {12}} (\bibinfo {year} {1972})}\BibitemShut {NoStop}%
\bibitem [{\citenamefont {Ramakrishna}\ and\ \citenamefont {Rabitz}(1996)}]{ramakrishna1996relation}%
  \BibitemOpen
  \bibfield  {author} {\bibinfo {author} {\bibfnamefont {V.}~\bibnamefont {Ramakrishna}}\ and\ \bibinfo {author} {\bibfnamefont {H.}~\bibnamefont {Rabitz}},\ }\bibfield  {title} {\bibinfo {title} {{Relation between quantum computing and quantum controllability}},\ }\href {https://doi.org/10.1103/PhysRevA.54.1715} {\bibfield  {journal} {\bibinfo  {journal} {Phys. Rev. A}\ }\textbf {\bibinfo {volume} {54}},\ \bibinfo {pages} {1715} (\bibinfo {year} {1996})}\BibitemShut {NoStop}%
\bibitem [{\citenamefont {Lloyd}(1995)}]{Lloyd_95}%
  \BibitemOpen
  \bibfield  {author} {\bibinfo {author} {\bibfnamefont {S.}~\bibnamefont {Lloyd}},\ }\bibfield  {title} {\bibinfo {title} {{Almost Any Quantum Logic Gate is Universal}},\ }\href {https://doi.org/10.1103/PhysRevLett.75.346} {\bibfield  {journal} {\bibinfo  {journal} {Phys. Rev. Lett.}\ }\textbf {\bibinfo {volume} {75}},\ \bibinfo {pages} {346} (\bibinfo {year} {1995})}\BibitemShut {NoStop}%
\bibitem [{\citenamefont {D’Alessandro}\ and\ \citenamefont {Hartwig}(2021)}]{d2021dynamical}%
  \BibitemOpen
  \bibfield  {author} {\bibinfo {author} {\bibfnamefont {D.}~\bibnamefont {D’Alessandro}}\ and\ \bibinfo {author} {\bibfnamefont {J.~T.}\ \bibnamefont {Hartwig}},\ }\bibfield  {title} {\bibinfo {title} {{Dynamical decomposition of bilinear control systems subject to symmetries}},\ }\href@noop {} {\bibfield  {journal} {\bibinfo  {journal} {Journal of Dynamical and Control Systems}\ }\textbf {\bibinfo {volume} {27}},\ \bibinfo {pages} {1} (\bibinfo {year} {2021})}\BibitemShut {NoStop}%
\bibitem [{\citenamefont {Wiersema}\ \emph {et~al.}(2023)\citenamefont {Wiersema}, \citenamefont {Kökcü}, \citenamefont {Kemper},\ and\ \citenamefont {Bakalov}}]{wiersema2023classification}%
  \BibitemOpen
  \bibfield  {author} {\bibinfo {author} {\bibfnamefont {R.}~\bibnamefont {Wiersema}}, \bibinfo {author} {\bibfnamefont {E.}~\bibnamefont {Kökcü}}, \bibinfo {author} {\bibfnamefont {A.~F.}\ \bibnamefont {Kemper}},\ and\ \bibinfo {author} {\bibfnamefont {B.~N.}\ \bibnamefont {Bakalov}},\ }\bibfield  {title} {\bibinfo {title} {{Classification of dynamical Lie algebras for translation-invariant 2-local spin systems in one dimension}},\ }\href {https://arxiv.org/abs/2309.05690} {\bibfield  {journal} {\bibinfo  {journal} {arXiv preprint arXiv:2309.05690}\ } (\bibinfo {year} {2023})}\BibitemShut {NoStop}%
\bibitem [{\citenamefont {Schirmer}\ \emph {et~al.}(2002)\citenamefont {Schirmer}, \citenamefont {Pullen},\ and\ \citenamefont {Solomon}}]{schirmer2002identification}%
  \BibitemOpen
  \bibfield  {author} {\bibinfo {author} {\bibfnamefont {S.}~\bibnamefont {Schirmer}}, \bibinfo {author} {\bibfnamefont {I.}~\bibnamefont {Pullen}},\ and\ \bibinfo {author} {\bibfnamefont {A.}~\bibnamefont {Solomon}},\ }\bibfield  {title} {\bibinfo {title} {{Identification of dynamical Lie algebras for finite-level quantum control systems}},\ }\href@noop {} {\bibfield  {journal} {\bibinfo  {journal} {Journal of Physics A: Mathematical and General}\ }\textbf {\bibinfo {volume} {35}},\ \bibinfo {pages} {2327} (\bibinfo {year} {2002})}\BibitemShut {NoStop}%
\bibitem [{\citenamefont {Raussendorf}\ \emph {et~al.}(2003)\citenamefont {Raussendorf}, \citenamefont {Browne},\ and\ \citenamefont {Briegel}}]{Raussendorf2003}%
  \BibitemOpen
  \bibfield  {author} {\bibinfo {author} {\bibfnamefont {R.}~\bibnamefont {Raussendorf}}, \bibinfo {author} {\bibfnamefont {D.~E.}\ \bibnamefont {Browne}},\ and\ \bibinfo {author} {\bibfnamefont {H.~J.}\ \bibnamefont {Briegel}},\ }\bibfield  {title} {\bibinfo {title} {{Measurement-based quantum computation on cluster states}},\ }\href {https://doi.org/10.1103/PhysRevA.68.022312} {\bibfield  {journal} {\bibinfo  {journal} {Phys. Rev. A}\ }\textbf {\bibinfo {volume} {68}},\ \bibinfo {pages} {022312} (\bibinfo {year} {2003})}\BibitemShut {NoStop}%
\bibitem [{\citenamefont {Raussendorf}\ \emph {et~al.}(2019)\citenamefont {Raussendorf}, \citenamefont {Okay}, \citenamefont {Wang}, \citenamefont {Stephen},\ and\ \citenamefont {Nautrup}}]{Raussendorf2019}%
  \BibitemOpen
  \bibfield  {author} {\bibinfo {author} {\bibfnamefont {R.}~\bibnamefont {Raussendorf}}, \bibinfo {author} {\bibfnamefont {C.}~\bibnamefont {Okay}}, \bibinfo {author} {\bibfnamefont {D.-S.}\ \bibnamefont {Wang}}, \bibinfo {author} {\bibfnamefont {D.~T.}\ \bibnamefont {Stephen}},\ and\ \bibinfo {author} {\bibfnamefont {H.~P.}\ \bibnamefont {Nautrup}},\ }\bibfield  {title} {\bibinfo {title} {{Computationally Universal Phase of Quantum Matter}},\ }\href {https://doi.org/10.1103/PhysRevLett.122.090501} {\bibfield  {journal} {\bibinfo  {journal} {Phys. Rev. Lett.}\ }\textbf {\bibinfo {volume} {122}},\ \bibinfo {pages} {090501} (\bibinfo {year} {2019})}\BibitemShut {NoStop}%
\bibitem [{\citenamefont {Devakul}\ and\ \citenamefont {Williamson}(2018)}]{Devakul2018}%
  \BibitemOpen
  \bibfield  {author} {\bibinfo {author} {\bibfnamefont {T.}~\bibnamefont {Devakul}}\ and\ \bibinfo {author} {\bibfnamefont {D.~J.}\ \bibnamefont {Williamson}},\ }\bibfield  {title} {\bibinfo {title} {{Universal quantum computation using fractal symmetry-protected cluster phases}},\ }\href {https://doi.org/10.1103/PhysRevA.98.022332} {\bibfield  {journal} {\bibinfo  {journal} {Phys. Rev. A}\ }\textbf {\bibinfo {volume} {98}},\ \bibinfo {pages} {022332} (\bibinfo {year} {2018})}\BibitemShut {NoStop}%
\bibitem [{\citenamefont {Stephen}\ \emph {et~al.}(2019)\citenamefont {Stephen}, \citenamefont {{Poulsen Nautrup}}, \citenamefont {Bermejo-Vega}, \citenamefont {Eisert},\ and\ \citenamefont {Raussendorf}}]{Stephen2019subsystem}%
  \BibitemOpen
  \bibfield  {author} {\bibinfo {author} {\bibfnamefont {D.~T.}\ \bibnamefont {Stephen}}, \bibinfo {author} {\bibfnamefont {H.}~\bibnamefont {{Poulsen Nautrup}}}, \bibinfo {author} {\bibfnamefont {J.}~\bibnamefont {Bermejo-Vega}}, \bibinfo {author} {\bibfnamefont {J.}~\bibnamefont {Eisert}},\ and\ \bibinfo {author} {\bibfnamefont {R.}~\bibnamefont {Raussendorf}},\ }\bibfield  {title} {\bibinfo {title} {{Subsystem symmetries, quantum cellular automata, and computational phases of quantum matter}},\ }\href {https://doi.org/10.22331/q-2019-05-20-142} {\bibfield  {journal} {\bibinfo  {journal} {{Quantum}}\ }\textbf {\bibinfo {volume} {3}},\ \bibinfo {pages} {142} (\bibinfo {year} {2019})}\BibitemShut {NoStop}%
\bibitem [{\citenamefont {Daniel}\ \emph {et~al.}(2020)\citenamefont {Daniel}, \citenamefont {Alexander},\ and\ \citenamefont {Miyake}}]{Daniel2020computational}%
  \BibitemOpen
  \bibfield  {author} {\bibinfo {author} {\bibfnamefont {A.~K.}\ \bibnamefont {Daniel}}, \bibinfo {author} {\bibfnamefont {R.~N.}\ \bibnamefont {Alexander}},\ and\ \bibinfo {author} {\bibfnamefont {A.}~\bibnamefont {Miyake}},\ }\bibfield  {title} {\bibinfo {title} {{Computational universality of symmetry-protected topologically ordered cluster phases on 2{D} {A}rchimedean lattices}},\ }\href {https://doi.org/10.22331/q-2020-02-10-228} {\bibfield  {journal} {\bibinfo  {journal} {{Quantum}}\ }\textbf {\bibinfo {volume} {4}},\ \bibinfo {pages} {228} (\bibinfo {year} {2020})}\BibitemShut {NoStop}%
\bibitem [{\citenamefont {Smith}\ \emph {et~al.}(2024)\citenamefont {Smith}, \citenamefont {{Poulsen Nautrup}},\ and\ \citenamefont {Briegel}}]{Parity_MBQC}%
  \BibitemOpen
  \bibfield  {author} {\bibinfo {author} {\bibfnamefont {I.~D.}\ \bibnamefont {Smith}}, \bibinfo {author} {\bibfnamefont {H.}~\bibnamefont {{Poulsen Nautrup}}},\ and\ \bibinfo {author} {\bibfnamefont {H.~J.}\ \bibnamefont {Briegel}},\ }\bibfield  {title} {\bibinfo {title} {{Parity Quantum Computing as $yz$-Plane Measurement-Based Quantum Computing}},\ }\href {https://doi.org/10.1103/PhysRevLett.132.220602} {\bibfield  {journal} {\bibinfo  {journal} {Phys. Rev. Lett.}\ }\textbf {\bibinfo {volume} {132}},\ \bibinfo {pages} {220602} (\bibinfo {year} {2024})}\BibitemShut {NoStop}%
\bibitem [{\citenamefont {Stephen}\ \emph {et~al.}(2024)\citenamefont {Stephen}, \citenamefont {Ho}, \citenamefont {Wei}, \citenamefont {Raussendorf},\ and\ \citenamefont {Verresen}}]{Stephen2024}%
  \BibitemOpen
  \bibfield  {author} {\bibinfo {author} {\bibfnamefont {D.~T.}\ \bibnamefont {Stephen}}, \bibinfo {author} {\bibfnamefont {W.~W.}\ \bibnamefont {Ho}}, \bibinfo {author} {\bibfnamefont {T.-C.}\ \bibnamefont {Wei}}, \bibinfo {author} {\bibfnamefont {R.}~\bibnamefont {Raussendorf}},\ and\ \bibinfo {author} {\bibfnamefont {R.}~\bibnamefont {Verresen}},\ }\bibfield  {title} {\bibinfo {title} {{Universal Measurement-Based Quantum Computation in a One-Dimensional Architecture Enabled by Dual-Unitary Circuits}},\ }\href {https://doi.org/10.1103/PhysRevLett.132.250601} {\bibfield  {journal} {\bibinfo  {journal} {Phys. Rev. Lett.}\ }\textbf {\bibinfo {volume} {132}},\ \bibinfo {pages} {250601} (\bibinfo {year} {2024})}\BibitemShut {NoStop}%
\bibitem [{\citenamefont {Raussendorf}(2005)}]{Raussendorf2005}%
  \BibitemOpen
  \bibfield  {author} {\bibinfo {author} {\bibfnamefont {R.}~\bibnamefont {Raussendorf}},\ }\bibfield  {title} {\bibinfo {title} {{Quantum computation via translation-invariant operations on a chain of qubits}},\ }\href {https://doi.org/10.1103/PhysRevA.72.052301} {\bibfield  {journal} {\bibinfo  {journal} {Phys. Rev. A}\ }\textbf {\bibinfo {volume} {72}},\ \bibinfo {pages} {052301} (\bibinfo {year} {2005})}\BibitemShut {NoStop}%
\bibitem [{\citenamefont {Poulsen~Nautrup}\ and\ \citenamefont {Briegel}(2024)}]{nautrup2023}%
  \BibitemOpen
  \bibfield  {author} {\bibinfo {author} {\bibfnamefont {H.}~\bibnamefont {Poulsen~Nautrup}}\ and\ \bibinfo {author} {\bibfnamefont {H.~J.}\ \bibnamefont {Briegel}},\ }\bibfield  {title} {\bibinfo {title} {Measurement-based quantum computation from clifford quantum cellular automata},\ }\href {https://doi.org/10.1103/PhysRevA.110.062617} {\bibfield  {journal} {\bibinfo  {journal} {Phys. Rev. A}\ }\textbf {\bibinfo {volume} {110}},\ \bibinfo {pages} {062617} (\bibinfo {year} {2024})}\BibitemShut {NoStop}%
\bibitem [{\citenamefont {Steudtner}\ and\ \citenamefont {Wehner}(2018)}]{steudtner2018fermion}%
  \BibitemOpen
  \bibfield  {author} {\bibinfo {author} {\bibfnamefont {M.}~\bibnamefont {Steudtner}}\ and\ \bibinfo {author} {\bibfnamefont {S.}~\bibnamefont {Wehner}},\ }\bibfield  {title} {\bibinfo {title} {{Fermion-to-qubit mappings with varying resource requirements for quantum simulation}},\ }\href@noop {} {\bibfield  {journal} {\bibinfo  {journal} {New Journal of Physics}\ }\textbf {\bibinfo {volume} {20}},\ \bibinfo {pages} {063010} (\bibinfo {year} {2018})}\BibitemShut {NoStop}%
\bibitem [{\citenamefont {Romero}\ \emph {et~al.}(2018)\citenamefont {Romero}, \citenamefont {Babbush}, \citenamefont {McClean}, \citenamefont {Hempel}, \citenamefont {Love},\ and\ \citenamefont {Aspuru-Guzik}}]{romero2018strategies}%
  \BibitemOpen
  \bibfield  {author} {\bibinfo {author} {\bibfnamefont {J.}~\bibnamefont {Romero}}, \bibinfo {author} {\bibfnamefont {R.}~\bibnamefont {Babbush}}, \bibinfo {author} {\bibfnamefont {J.~R.}\ \bibnamefont {McClean}}, \bibinfo {author} {\bibfnamefont {C.}~\bibnamefont {Hempel}}, \bibinfo {author} {\bibfnamefont {P.~J.}\ \bibnamefont {Love}},\ and\ \bibinfo {author} {\bibfnamefont {A.}~\bibnamefont {Aspuru-Guzik}},\ }\bibfield  {title} {\bibinfo {title} {{Strategies for quantum computing molecular energies using the unitary coupled cluster ansatz}},\ }\href@noop {} {\bibfield  {journal} {\bibinfo  {journal} {Quantum Science and Technology}\ }\textbf {\bibinfo {volume} {4}},\ \bibinfo {pages} {014008} (\bibinfo {year} {2018})}\BibitemShut {NoStop}%
\bibitem [{\citenamefont {Cowtan}\ \emph {et~al.}(2020)\citenamefont {Cowtan}, \citenamefont {Simmons},\ and\ \citenamefont {Duncan}}]{cowtan2020generic}%
  \BibitemOpen
  \bibfield  {author} {\bibinfo {author} {\bibfnamefont {A.}~\bibnamefont {Cowtan}}, \bibinfo {author} {\bibfnamefont {W.}~\bibnamefont {Simmons}},\ and\ \bibinfo {author} {\bibfnamefont {R.}~\bibnamefont {Duncan}},\ }\bibfield  {title} {\bibinfo {title} {{A generic compilation strategy for the unitary coupled cluster ansatz}},\ }\href@noop {} {\bibfield  {journal} {\bibinfo  {journal} {arXiv preprint arXiv:2007.10515}\ } (\bibinfo {year} {2020})}\BibitemShut {NoStop}%
\bibitem [{\citenamefont {Van Den~Berg}\ and\ \citenamefont {Temme}(2020)}]{van2020circuit}%
  \BibitemOpen
  \bibfield  {author} {\bibinfo {author} {\bibfnamefont {E.}~\bibnamefont {Van Den~Berg}}\ and\ \bibinfo {author} {\bibfnamefont {K.}~\bibnamefont {Temme}},\ }\bibfield  {title} {\bibinfo {title} {{Circuit optimization of Hamiltonian simulation by simultaneous diagonalization of Pauli clusters}},\ }\href@noop {} {\bibfield  {journal} {\bibinfo  {journal} {Quantum}\ }\textbf {\bibinfo {volume} {4}},\ \bibinfo {pages} {322} (\bibinfo {year} {2020})}\BibitemShut {NoStop}%
\bibitem [{\citenamefont {Debnath}\ \emph {et~al.}(2016)\citenamefont {Debnath}, \citenamefont {Linke}, \citenamefont {Figgatt}, \citenamefont {Landsman}, \citenamefont {Wright},\ and\ \citenamefont {Monroe}}]{debnath2016demonstration}%
  \BibitemOpen
  \bibfield  {author} {\bibinfo {author} {\bibfnamefont {S.}~\bibnamefont {Debnath}}, \bibinfo {author} {\bibfnamefont {N.~M.}\ \bibnamefont {Linke}}, \bibinfo {author} {\bibfnamefont {C.}~\bibnamefont {Figgatt}}, \bibinfo {author} {\bibfnamefont {K.~A.}\ \bibnamefont {Landsman}}, \bibinfo {author} {\bibfnamefont {K.}~\bibnamefont {Wright}},\ and\ \bibinfo {author} {\bibfnamefont {C.}~\bibnamefont {Monroe}},\ }\bibfield  {title} {\bibinfo {title} {{Demonstration of a small programmable quantum computer with atomic qubits}},\ }\href@noop {} {\bibfield  {journal} {\bibinfo  {journal} {Nature}\ }\textbf {\bibinfo {volume} {536}},\ \bibinfo {pages} {63} (\bibinfo {year} {2016})}\BibitemShut {NoStop}%
\bibitem [{\citenamefont {Nam}\ \emph {et~al.}(2020)\citenamefont {Nam}, \citenamefont {Chen}, \citenamefont {Pisenti}, \citenamefont {Wright}, \citenamefont {Delaney}, \citenamefont {Maslov}, \citenamefont {Brown}, \citenamefont {Allen}, \citenamefont {Amini}, \citenamefont {Apisdorf} \emph {et~al.}}]{nam2020ground}%
  \BibitemOpen
  \bibfield  {author} {\bibinfo {author} {\bibfnamefont {Y.}~\bibnamefont {Nam}}, \bibinfo {author} {\bibfnamefont {J.-S.}\ \bibnamefont {Chen}}, \bibinfo {author} {\bibfnamefont {N.~C.}\ \bibnamefont {Pisenti}}, \bibinfo {author} {\bibfnamefont {K.}~\bibnamefont {Wright}}, \bibinfo {author} {\bibfnamefont {C.}~\bibnamefont {Delaney}}, \bibinfo {author} {\bibfnamefont {D.}~\bibnamefont {Maslov}}, \bibinfo {author} {\bibfnamefont {K.~R.}\ \bibnamefont {Brown}}, \bibinfo {author} {\bibfnamefont {S.}~\bibnamefont {Allen}}, \bibinfo {author} {\bibfnamefont {J.~M.}\ \bibnamefont {Amini}}, \bibinfo {author} {\bibfnamefont {J.}~\bibnamefont {Apisdorf}}, \emph {et~al.},\ }\bibfield  {title} {\bibinfo {title} {{Ground-state energy estimation of the water molecule on a trapped-ion quantum computer}},\ }\href@noop {} {\bibfield  {journal} {\bibinfo  {journal} {npj Quantum Information}\ }\textbf {\bibinfo {volume} {6}},\ \bibinfo {pages} {33} (\bibinfo {year} {2020})}\BibitemShut {NoStop}%
\bibitem [{\citenamefont {Klaver}\ \emph {et~al.}(2024)\citenamefont {Klaver}, \citenamefont {Rombouts}, \citenamefont {Fellner}, \citenamefont {Messinger}, \citenamefont {Ender}, \citenamefont {Ludwig},\ and\ \citenamefont {Lechner}}]{klaver2024swapless}%
  \BibitemOpen
  \bibfield  {author} {\bibinfo {author} {\bibfnamefont {B.}~\bibnamefont {Klaver}}, \bibinfo {author} {\bibfnamefont {S.}~\bibnamefont {Rombouts}}, \bibinfo {author} {\bibfnamefont {M.}~\bibnamefont {Fellner}}, \bibinfo {author} {\bibfnamefont {A.}~\bibnamefont {Messinger}}, \bibinfo {author} {\bibfnamefont {K.}~\bibnamefont {Ender}}, \bibinfo {author} {\bibfnamefont {K.}~\bibnamefont {Ludwig}},\ and\ \bibinfo {author} {\bibfnamefont {W.}~\bibnamefont {Lechner}},\ }\href {https://arxiv.org/abs/2408.10907} {\bibinfo {title} {{SWAP-less Implementation of Quantum Algorithms}}} (\bibinfo {year} {2024})\BibitemShut {NoStop}%
\bibitem [{\citenamefont {Bombín}(2016)}]{bombin2016dimensional}%
  \BibitemOpen
  \bibfield  {author} {\bibinfo {author} {\bibfnamefont {H.}~\bibnamefont {Bombín}},\ }\bibfield  {title} {\bibinfo {title} {{Dimensional jump in quantum error correction}},\ }\href {https://doi.org/10.1088/1367-2630/18/4/043038} {\bibfield  {journal} {\bibinfo  {journal} {New Journal of Physics}\ }\textbf {\bibinfo {volume} {18}},\ \bibinfo {pages} {043038} (\bibinfo {year} {2016})}\BibitemShut {NoStop}%
\bibitem [{sup()}]{supp}%
  \BibitemOpen
  \href@noop {} {}\bibinfo {note} {The Supplemental Material at URL-will-be-inserted-by-publisher contains several proofs of results presented in the main text and includes references \cite{nielsen2010quantum,oeis}.}\BibitemShut {Stop}%
\bibitem [{\citenamefont {Hall}\ and\ \citenamefont {Hall}(2013)}]{hall2013lie}%
  \BibitemOpen
  \bibfield  {author} {\bibinfo {author} {\bibfnamefont {B.~C.}\ \bibnamefont {Hall}}\ and\ \bibinfo {author} {\bibfnamefont {B.~C.}\ \bibnamefont {Hall}},\ }\href@noop {} {\emph {\bibinfo {title} {{Lie groups, Lie algebras, and representations}}}}\ (\bibinfo  {publisher} {Springer},\ \bibinfo {year} {2013})\BibitemShut {NoStop}%
\bibitem [{\citenamefont {Van~den Nest}\ \emph {et~al.}(2004)\citenamefont {Van~den Nest}, \citenamefont {Dehaene},\ and\ \citenamefont {De~Moor}}]{Nest_04}%
  \BibitemOpen
  \bibfield  {author} {\bibinfo {author} {\bibfnamefont {M.}~\bibnamefont {Van~den Nest}}, \bibinfo {author} {\bibfnamefont {J.}~\bibnamefont {Dehaene}},\ and\ \bibinfo {author} {\bibfnamefont {B.}~\bibnamefont {De~Moor}},\ }\bibfield  {title} {\bibinfo {title} {{Graphical description of the action of local Clifford transformations on graph states}},\ }\href {https://doi.org/10.1103/PhysRevA.69.022316} {\bibfield  {journal} {\bibinfo  {journal} {Phys. Rev. A}\ }\textbf {\bibinfo {volume} {69}},\ \bibinfo {pages} {022316} (\bibinfo {year} {2004})}\BibitemShut {NoStop}%
\bibitem [{\citenamefont {Aaronson}\ and\ \citenamefont {Gottesman}(2004)}]{Aaronson_04}%
  \BibitemOpen
  \bibfield  {author} {\bibinfo {author} {\bibfnamefont {S.}~\bibnamefont {Aaronson}}\ and\ \bibinfo {author} {\bibfnamefont {D.}~\bibnamefont {Gottesman}},\ }\bibfield  {title} {\bibinfo {title} {{Improved simulation of stabilizer circuits}},\ }\href {https://doi.org/10.1103/PhysRevA.70.052328} {\bibfield  {journal} {\bibinfo  {journal} {Phys. Rev. A}\ }\textbf {\bibinfo {volume} {70}},\ \bibinfo {pages} {052328} (\bibinfo {year} {2004})}\BibitemShut {NoStop}%
\bibitem [{Note1()}]{Note1}%
  \BibitemOpen
  \bibinfo {note} {Note the difference in the allowed set of scalars compared to that for adjoint universality}\BibitemShut {NoStop}%
\bibitem [{Note2()}]{Note2}%
  \BibitemOpen
  \bibinfo {note} {Example 1 does not follow the curve since it has a non-uniform distribution of anticommutation; some generators anticommute with half the others while some only anticommute with a constant number of others.}\BibitemShut {Stop}%
\bibitem [{\citenamefont {Bäumer}\ and\ \citenamefont {Woerner}(2024)}]{baeumer2024measurement}%
  \BibitemOpen
  \bibfield  {author} {\bibinfo {author} {\bibfnamefont {E.}~\bibnamefont {Bäumer}}\ and\ \bibinfo {author} {\bibfnamefont {S.}~\bibnamefont {Woerner}},\ }\bibfield  {title} {\bibinfo {title} {{Measurement-Based Long-Range Entangling Gates in Constant Depth}},\ }\href {https://arxiv.org/abs/2408.03064} {\bibfield  {journal} {\bibinfo  {journal} {arXiv preprint arXiv:2408.03064}\ } (\bibinfo {year} {2024})}\BibitemShut {NoStop}%
\bibitem [{\citenamefont {Poulin}(2005)}]{poulin2005stabilizer}%
  \BibitemOpen
  \bibfield  {author} {\bibinfo {author} {\bibfnamefont {D.}~\bibnamefont {Poulin}},\ }\bibfield  {title} {\bibinfo {title} {{Stabilizer Formalism for Operator Quantum Error Correction}},\ }\href {https://doi.org/10.1103/PhysRevLett.95.230504} {\bibfield  {journal} {\bibinfo  {journal} {Phys. Rev. Lett.}\ }\textbf {\bibinfo {volume} {95}},\ \bibinfo {pages} {230504} (\bibinfo {year} {2005})}\BibitemShut {NoStop}%
\bibitem [{\citenamefont {Paetznick}\ and\ \citenamefont {Reichardt}(2013)}]{paetznick2013universal}%
  \BibitemOpen
  \bibfield  {author} {\bibinfo {author} {\bibfnamefont {A.}~\bibnamefont {Paetznick}}\ and\ \bibinfo {author} {\bibfnamefont {B.~W.}\ \bibnamefont {Reichardt}},\ }\bibfield  {title} {\bibinfo {title} {{Universal Fault-Tolerant Quantum Computation with Only Transversal Gates and Error Correction}},\ }\href {https://doi.org/10.1103/PhysRevLett.111.090505} {\bibfield  {journal} {\bibinfo  {journal} {Phys. Rev. Lett.}\ }\textbf {\bibinfo {volume} {111}},\ \bibinfo {pages} {090505} (\bibinfo {year} {2013})}\BibitemShut {NoStop}%
\bibitem [{\citenamefont {Bombin}\ and\ \citenamefont {Martin-Delgado}(2006)}]{bombin2006topological}%
  \BibitemOpen
  \bibfield  {author} {\bibinfo {author} {\bibfnamefont {H.}~\bibnamefont {Bombin}}\ and\ \bibinfo {author} {\bibfnamefont {M.~A.}\ \bibnamefont {Martin-Delgado}},\ }\bibfield  {title} {\bibinfo {title} {{Topological Quantum Distillation}},\ }\href {https://doi.org/10.1103/PhysRevLett.97.180501} {\bibfield  {journal} {\bibinfo  {journal} {Phys. Rev. Lett.}\ }\textbf {\bibinfo {volume} {97}},\ \bibinfo {pages} {180501} (\bibinfo {year} {2006})}\BibitemShut {NoStop}%
\bibitem [{\citenamefont {Beverland}\ \emph {et~al.}(2021)\citenamefont {Beverland}, \citenamefont {Kubica},\ and\ \citenamefont {Svore}}]{beverland2021cost}%
  \BibitemOpen
  \bibfield  {author} {\bibinfo {author} {\bibfnamefont {M.~E.}\ \bibnamefont {Beverland}}, \bibinfo {author} {\bibfnamefont {A.}~\bibnamefont {Kubica}},\ and\ \bibinfo {author} {\bibfnamefont {K.~M.}\ \bibnamefont {Svore}},\ }\bibfield  {title} {\bibinfo {title} {{Cost of Universality: A Comparative Study of the Overhead of State Distillation and Code Switching with Color Codes}},\ }\href {https://doi.org/10.1103/PRXQuantum.2.020341} {\bibfield  {journal} {\bibinfo  {journal} {PRX Quantum}\ }\textbf {\bibinfo {volume} {2}},\ \bibinfo {pages} {020341} (\bibinfo {year} {2021})}\BibitemShut {NoStop}%
\bibitem [{\citenamefont {Lucas}(2014)}]{lucas2014ising}%
  \BibitemOpen
  \bibfield  {author} {\bibinfo {author} {\bibfnamefont {A.}~\bibnamefont {Lucas}},\ }\bibfield  {title} {\bibinfo {title} {{Ising formulations of many NP problems}},\ }\href@noop {} {\bibfield  {journal} {\bibinfo  {journal} {Frontiers in physics}\ }\textbf {\bibinfo {volume} {2}},\ \bibinfo {pages} {5} (\bibinfo {year} {2014})}\BibitemShut {NoStop}%
\bibitem [{\citenamefont {Farhi}\ \emph {et~al.}(2014)\citenamefont {Farhi}, \citenamefont {Goldstone},\ and\ \citenamefont {Gutmann}}]{farhi2014quantumapproximateoptimizationalgorithm}%
  \BibitemOpen
  \bibfield  {author} {\bibinfo {author} {\bibfnamefont {E.}~\bibnamefont {Farhi}}, \bibinfo {author} {\bibfnamefont {J.}~\bibnamefont {Goldstone}},\ and\ \bibinfo {author} {\bibfnamefont {S.}~\bibnamefont {Gutmann}},\ }\href {https://arxiv.org/abs/1411.4028} {\bibinfo {title} {{A Quantum Approximate Optimization Algorithm}}} (\bibinfo {year} {2014}),\ \Eprint {https://arxiv.org/abs/1411.4028} {arXiv:1411.4028 [quant-ph]} \BibitemShut {NoStop}%
\bibitem [{\citenamefont {Katz}\ \emph {et~al.}(2022)\citenamefont {Katz}, \citenamefont {Cetina},\ and\ \citenamefont {Monroe}}]{Katz2022}%
  \BibitemOpen
  \bibfield  {author} {\bibinfo {author} {\bibfnamefont {O.}~\bibnamefont {Katz}}, \bibinfo {author} {\bibfnamefont {M.}~\bibnamefont {Cetina}},\ and\ \bibinfo {author} {\bibfnamefont {C.}~\bibnamefont {Monroe}},\ }\bibfield  {title} {\bibinfo {title} {{$N$-Body Interactions between Trapped Ion Qubits via Spin-Dependent Squeezing}},\ }\href {https://doi.org/10.1103/PhysRevLett.129.063603} {\bibfield  {journal} {\bibinfo  {journal} {Phys. Rev. Lett.}\ }\textbf {\bibinfo {volume} {129}},\ \bibinfo {pages} {063603} (\bibinfo {year} {2022})}\BibitemShut {NoStop}%
\bibitem [{\citenamefont {Aguilar}\ \emph {et~al.}(2024)\citenamefont {Aguilar}, \citenamefont {Cichy}, \citenamefont {Eisert},\ and\ \citenamefont {Bittel}}]{aguilar2024classificationpauliliealgebras}%
  \BibitemOpen
  \bibfield  {author} {\bibinfo {author} {\bibfnamefont {G.}~\bibnamefont {Aguilar}}, \bibinfo {author} {\bibfnamefont {S.}~\bibnamefont {Cichy}}, \bibinfo {author} {\bibfnamefont {J.}~\bibnamefont {Eisert}},\ and\ \bibinfo {author} {\bibfnamefont {L.}~\bibnamefont {Bittel}},\ }\href {https://arxiv.org/abs/2408.00081} {\bibinfo {title} {{Full classification of Pauli Lie algebras}}} (\bibinfo {year} {2024}),\ \Eprint {https://arxiv.org/abs/2408.00081} {arXiv:2408.00081 [quant-ph]} \BibitemShut {NoStop}%
\bibitem [{\citenamefont {Nielsen}\ and\ \citenamefont {Chuang}(2010)}]{nielsen2010quantum}%
  \BibitemOpen
  \bibfield  {author} {\bibinfo {author} {\bibfnamefont {M.~A.}\ \bibnamefont {Nielsen}}\ and\ \bibinfo {author} {\bibfnamefont {I.~L.}\ \bibnamefont {Chuang}},\ }\href@noop {} {\emph {\bibinfo {title} {{Quantum computation and quantum information}}}}\ (\bibinfo  {publisher} {Cambridge university press},\ \bibinfo {year} {2010})\BibitemShut {NoStop}%
\bibitem [{\citenamefont {{OEIS Foundation Inc.}}(2024)}]{oeis}%
  \BibitemOpen
  \bibfield  {author} {\bibinfo {author} {\bibnamefont {{OEIS Foundation Inc.}}},\ }\href@noop {} {\bibinfo {title} {{The {O}n-{L}ine {E}ncyclopedia of {I}nteger {S}equences}}} (\bibinfo {year} {2024}),\ \bibinfo {note} {published electronically at \url{http://oeis.org/A143441}}\BibitemShut {NoStop}%
\end{thebibliography}%

\newpage

\appendix
\onecolumngrid

\section{BCH Formula for Pauli Strings} \label{app:BCH_paulis}

Here, we spell out the details regarding the reduction of the BCH formula for elements of $i\mathcal{P}_{N}^{*}$ expressed in Eq.~(4)  in the main text. 

Let $A, B \in i\mathcal{P}_{N}^{*}$ be such that $[A,B] \neq 0$, i.e. that $[A,B] = 2AB$ or equivalently $AB = -BA$. It follows that 
\begin{align}
[A,[A,B]] = 2[A,AB] = 2(A^{2}B - ABA) = 4A^{2}B = -4B 
\end{align} 
where the last equality uses the fact that elements of $i\mathcal{P}_{N}^{*}$ square to minus the identity. Using the BCH formula, we have
\begin{align}
e^{\frac{\pi}{4}A}e^{tB}e^{-\frac{\pi}{4}A} = e^{tB + \sum_{k = 1}^{\infty}\frac{1}{k!} ad_{\frac{\pi}{4}A}^{(k)}(tB)}
\end{align}
where $ad_{\frac{\pi}{4}A}^{(k)}(tB)$ denotes the composition of $ad_{\frac{\pi}{4}A}$ $k$-times and applied to $tB$, as in the main text. Using the above commutation relations, we see that 
\begin{align}
tB + \sum_{k = 1}^{\infty}\frac{1}{k!} ad_{\frac{\pi}{4}A}^{(k)}(tB) &= t \left\{B + \sum_{k = 1}^{\infty}\frac{1}{k!}\left(\frac{\pi}{4}\right)^{k} ad_{A}^{(k)}(B) \right\} \\
&= t \left\{B\left(\sum_{k=0}^{\infty} \frac{(-1)^{k}}{(2k)!}\left(\frac{\pi}{4}\right)^{2k}2^{2k} \right) + [A,B]\left(\sum_{k=0}^{\infty}\frac{(-1)^{k}}{(2k+1)!}\left(\frac{\pi}{4}\right)^{2k}2^{2k} \right)\right\} \\
&= t\left(\cos\left(\frac{\pi}{2}\right)B + \frac{1}{2}\sin\left(\frac{\pi}{2}\right)[A,B] \right) \\
&= \frac{t}{2}[A,B].
\end{align}

\section{The \textsc{PauliCompiler} and the proof of Theorem~1}\label{appendix:thm:construction}

\begin{figure}[ht!]
\begin{minipage}{\linewidth} 
\begin{algorithm}[H]
\caption{\textsc{PauliCompiler}$(P)$}
\begin{algorithmic}[1]
\Require Target: $P=iV\otimes W\in i\mathcal{P}^*_{N}$ for $V\in\mathcal{P}_k,W\in\mathcal{P}_{N-k}$.\\ \hspace{1.2cm} $\mathcal{A},\mathcal{A}',\mathcal{B},\mathcal{B}'$ as in Theorem~\ref{thm:app:construction}. \vspace{0.2cm}
\If{$W=I$}
\State Choose $A_1,...,A_s\in \mathcal{A}$ s.t. $iV\propto ad_{A_1}\cdots ad_{A_{s-1}}(A_s)\wedge s=\mathcal{O}(1)$\Comment{Use adjoint universality of $\mathcal{A}$.}\\
\Return $(A_1\otimes I,...,A_s\otimes I)$ 
\ElsIf{$V=I$}
\State Choose $W_1, W_2\in i\mathcal{P}_{N-k}$ s.t. $iW\propto [W_1,W_2]$
\State \begin{varwidth}[t]{\linewidth}
$\mathcal{G}'=(G_1',...,G_{\mathcal{G}'}')\gets\textsc{SubsystemCompiler}(W_1)$\par \hspace{0.7cm} s.t. $V_1'\otimes W_1\propto ad_{G'_{|\mathcal{G}'|}}\cdots ad_{G_2 '}(G_1')$
\end{varwidth}\Comment{Construct $V_1'\otimes W_1$ for a given $W_1$ and some $V_1'\in \mathcal{P}^*_k$}
\State \begin{varwidth}[t]{\linewidth}
$\mathcal{G}''=(G_1'',...,G_{\mathcal{G}''}'')\gets\textsc{SubsystemCompiler}(W_2)$\par \hspace{0.7cm} s.t. $V_2'\otimes W_2\propto ad_{G''_{|\mathcal{G}''|}}\cdots ad_{G_2''}(G_1'')$
\end{varwidth}\Comment{Construct $V_2'\otimes W_2$ for a given $W_2$ and some $V_2'\in \mathcal{P}^*_k$}
\State Choose $A_1,...,A_s\in\mathcal{A}$ s.t. $iV_1'\propto ad_{A_1}\cdots ad_{A_s}(iV_2')\wedge s=\mathcal{O}(1)$ \Comment{Use Proposition~\ref{prop:app:Vmap}.}
\State $\mathcal{G}=(G_1,...,G_{|\mathcal{G}|})\equiv\mathcal{G}'\oplus \mathcal{G}''\oplus (A_1\otimes I,...,A_s\otimes I)$\Comment{Construct $V_1'\otimes W_1,V_1'\otimes W_2$ s.t. $[V_1'\otimes W_1,V_1'\otimes W_2]\propto I\otimes W$.}
\State Choose $\sigma\in S_{|\mathcal{G}|}$ s.t. $ad_{G_{\sigma(1)}}\cdots ad_{G_{\sigma(|\mathcal{G}|-1)}}(G_{\sigma(|\mathcal{G}|)})\neq 0$\Comment{Use Lemma~\ref{lemma:reordering}}\\
\Return $(G_{\sigma(1)},...,G_{\sigma(|\mathcal{G}|)})$
\Else
\State \begin{varwidth}[t]{\linewidth}
$\mathcal{G}'=(G_1',...,G_{\mathcal{G}'}')\gets\textsc{SubsystemCompiler}(iW)$\par \hspace{0.7cm} s.t. $iV'\otimes W\propto ad_{G'_{|\mathcal{G}'|}}\cdots ad_{G_2 '}(G_1')$
\end{varwidth}\Comment{Construct $V'\otimes W$ for a given $W$ and some $V'\in \mathcal{P}^*_k$}
\State Choose $A_1,...,A_s\in\mathcal{A}$ s.t. $iV\propto ad_{A_1}\cdots ad_{A_s}(iV')\wedge s=\mathcal{O}(1)$\Comment{Use Proposition~\ref{prop:app:Vmap}.}\\
\Return $({A_1\otimes I,...,A_s\otimes I,G'_{|\mathcal{G}'|},...,G'_1})$
\EndIf
\end{algorithmic}
\end{algorithm}
% \vspace{-1.75\baselineskip}
\end{minipage}
\caption{Algorithm~1 from the main text with additional line numbering and comments. Here, we write $\propto$ to indicate proportionality but nonzero and $A \oplus B$ is the concatenation of two sequences $A,B$.}
\label{fig:app-paulicompiler}
\end{figure}

For convenience, we restate Theorem~1 here:
\begin{Theorem}\label{thm:app:construction}
    Fix $k \geq 2$ and consider $\mathcal{A} \subset i\mathcal{P}_k^*$  and $\mathcal{B}\subset i\mathcal{P}^*_{N-k}$ such that $\langle\mathcal{A}\rangle_{\mathrm{Lie}}=\mathfrak{su}(2^k)$ and $\mathcal{B}$ is product universal for $i\mathcal{P}^*_{N-k}$. Let $\mathbb{B}$ be a set of pairs $\{(U_{B}, B)\}_{B \in \mathcal{B}}$ where each $U_{B} \in \mathcal{P}_{k}^{*}$. Defining $\mathcal{A}'=\{A\otimes I^{\otimes N-k}|A\in\mathcal{A}\}$ and $\mathcal{B}'=\{U_{B} \otimes B | (U_{B}, B) \in \mathbb{B}\}$, we have that
    \begin{align*}
           \langle \mathcal{A}'\cup\mathcal{B}'\rangle_{\mathrm{Lie}}=\mathfrak{su}(2^N).
    \end{align*}
\end{Theorem}

The proof of this theorem is established by the correctness of the algorithm \textsc{PauliCompiler} presented in the main text and depicted again here in \Cref{fig:app-paulicompiler}. We state and prove the correctness of the algorithm as a lemma below, but before doing so, let us analyze the algorithm a bit more closely to develop an intuition for what is doing. To aid the following discussion, the algorithm as presented in \Cref{fig:app-paulicompiler} contains additional line numbering to the algorithm presented in the main text.

The algorithm relies on sets of operators $\mathcal{A}',\mathcal{B}'\subset i\mathcal{P}_N$ constructed from $\mathcal{A}\subset i\mathcal{P}^*_k$ and $\mathcal{B}\subset i\mathcal{P}_{N-k}$ as defined in Theorem~\ref{thm:app:construction}. Their particular structure is crucial for the workings of the algorithm. In particular, while $\mathcal{B}$ is product universal for $i\mathcal{P}_{N-k}$, each element of $\mathcal{B}'$ must also act non-trivially on the first $k$-qubit subspace. Only because of that can we make use of the adjoint universality of $\mathcal{A}$ to promote $\mathcal{A}'\cup \mathcal{B}'$ to adjoint universality on the full space using \textsc{SubsystemCompiler} of \Cref{appendix:lemma:alg} (cf.Fig.~\ref{fig:app-alg}). 

Given the input $P = iV\otimes W\in i\mathcal{P}^*_N$ where $V \in \mathcal{P}_{k}$ and $W \in \mathcal{P}_{N-k}$, \textsc{PauliCompiler} distinguishes three cases:
\begin{itemize}
    \item \textbf{Case 1:} $W=I$ (lines 2-4). Here, we make use of the fact that $\mathcal{A}$ is adjoint universal to produce any operator $iV\otimes I$ for $V\in \mathcal{P}_k$.
    \item \textbf{Case 2:} $W\neq I$ and $V\neq I$ (lines 13-16). In line 14, we rely on \textsc{SubsystemCompiler}$(iW)$ to compile a sequence $iV'\otimes W$ for some $V'\in \mathcal{P}^*_{k}$. Given Proposition~\ref{prop:app:Vmap}, we can then map this to an arbitrary $iV$ (up to a factor).
    \item \textbf{Case 3:} $V=I$ (lines 5-12). This is the most complicated case. Instead of generating $I\otimes W$ directly, we compile two $V_1'\otimes W_1$ and $V_1'\otimes W_2$ (using the same procedure as in case 2) such that $[V_1'\otimes W_1, V_1'\otimes W_2]\propto I\otimes W$. However, it is not obvious that for two operators $B=ad_{B_1}\cdots ad_{B_{r-1}}(B_r)$ and $C=ad_{C_1}\cdots ad_{C_{l-1}}(C_l)$, we can write $[B,C]=ad_{D_1}\cdots ad_{D_{m-1}}(D_m)$. As we prove in Lemma~\ref{lemma:reordering}, we can reorder the elements (as in line 11) such that the commutator of commutators is a nested commutator. 
\end{itemize}
To implement the algorithm, we write a table for $\mathcal{A}$ that lists for each $iV\in i\mathcal{P}_k$ a sequence $A_1,...,A_s$ as in line 3. Then, one writes another table that lists for each pair $V,V'\in\mathcal{P}_k$ a sequence $A_1,...,A_s$ as in lines 9 and 15. To generate the nested commutator from the commutators of commutators in line 11 (and implicitly in lines 9 and 15 due to Proposition~\ref{prop:app:Vmap}), one iteratively applies Lemma~\ref{lemma:reordering}. \textsc{SubsystemCompiler} can be implemented directly as described in the pseudocode. We will see in Proposition~\ref{prop:time} in Appendix~\ref{app:time} that the algorithm has $\mathcal{O}(N^2)$ time complexity.

\begin{Lemma}\label{lemma:alg:paulicompiler}
    For any input  $P\in i\mathcal{P}^*_{N}$ with $P=iV\otimes W$ for $V\in\mathcal{P}_k,W\in\mathcal{P}_{N-k}$ and sets $\mathcal{A}',\mathcal{B}'$ as in Theorem~\ref{thm:app:construction}, Algorithm~\ref{fig:app-paulicompiler} finishes and returns a finite, ordered set of operators $(G_1,...,G_{L})$ with $G_i\in \mathcal{A}'\cup \mathcal{B}'$ for all $i=1,...,L$ such that $\exists\alpha\in\mathbb{R}^*$ with $ad_{G_{1}}\cdots ad_{G_{L-1}}(G_L)= \alpha\: P$.
\end{Lemma}
\begin{proof}
    The algorithm is depicted in Fig.~\ref{fig:app-paulicompiler} with additional line numbering which will be referred to in this proof. This algorithm relies on the algorithm \textsc{SubsystemCompiler}, the correctness of which is established below in \Cref{appendix:lemma:alg}. Here, we rely on it to return a sequence of operators whose nested commutator is proportional to $iV'\otimes W$ given $W\in \mathcal{P}^*_{N-k}$ for some $V'\in i\mathcal{P}^*_k$.\\

    \textbf{Statement 1:} \emph{The algorithm is correct.} We show that the algorithm is correct for any input through the four following statements.\\

    \textbf{Statement 1.1:} \emph{All choices in the algorithm are valid.} There are five choices in the algorithms, in lines 3, 6, 9, 11 and 15.

    Line 3: Since by construction $\mathcal{A}\subset i\mathcal{P}_k$ is adjoint universal, for each $V\in \mathcal{P}_k$ there exist a finite sequence $A_1,...,A_s\in\mathcal{A}$ such that $iV\propto ad_{A_1}\cdots ad_{A_{s-1}}(A_s)$ is nonzero, with $s=\mathcal{O}(f(k))$ for some function $f$ of $k$. Since $k$ is constant, so is $s$.

    Line 6. For any Pauli operator $W\in\mathcal{P}^*_{N-k}$ there exist two other Pauli operators $W_1,W_2\in\mathcal{P}_{N-k}$ that anticommute and $W_1W_2\propto W$. This is trivially true for $N-k=1$ and easily extends to all other $N-k$ by induction. The fact that $W\neq I$ is ensured by the \textsc{if} statement in line 2.

    Line 9 and 15. Proven by Propostion~\ref{prop:app:Vmap} in Appendix~\ref{app:prop:Vmap}. From the proof of Proposition~\ref{prop:app:Vmap}, we can immediately infer that the chosen sequences may be finite and have length $s=\mathcal{O}(f(k))$ for some function $f$ of $k$. Since $k$ is constant, so is $s$.

    Line 11. By construction, we have 
    \begin{align}
        I\otimes iW&\propto  [V_1'\otimes W_1,V_1'\otimes W_2]\nonumber\\
        &\propto[V_1'\otimes W_1, ad_{A_1}\cdots ad_{A_s}(V_2')\otimes W_2]\nonumber\\
        &\propto [ad_{G'_{|\mathcal{G}'|}}\cdots ad_{G_2 '}(G_1') ,ad_{A_1\otimes I}\cdots ad_{A_s\otimes I}ad_{G''_{|\mathcal{G}''|}}\cdots ad_{G_2''}(G_1'')]\\
        &= ad_{ad_{G'_{|\mathcal{G}'|}}\cdots ad_{G_2 '}(G_1')}\left(ad_{A_1\otimes I}\cdots ad_{A_s\otimes I}ad_{G''_{|\mathcal{G}''|}}\cdots ad_{G_2''}(G_1'')\right)\neq 0\nonumber
    \end{align}
    This is a commutator of commutators instead of the nested commutator.
    As we prove in Lemma~\ref{lemma:reordering} in Appendix~\ref{app:prop:Vmap}, any nonzero commutator of commutators can be written as a nonzero nested commutator with reordered elements.\\

    \textbf{Statement 1.2:} \emph{The returned sequence $G_1,...,G_L$ is such that $G_1,...,G_L\in\mathcal{A}'\cup\mathcal{B}'$.} This is true by construction of the algorithm as the elements returned in lines 4, 12 and 16 have all been selected from $\mathcal{A}'$ or $\mathcal{B}'$ in lines 3, 9 and 15 while we have proven in Lemma~\ref{appendix:lemma:alg} that the sequence returned by \textsc{SubsystemCompiler} only contains elements from the same sets.\\
    
    \textbf{Statement 1.3:} \emph{The returned sequence $G_1,...,G_L$ is such that $ad_{G_1}\cdots ad_{G_{L-1}}(G_L)\neq 0$.} This is true by construction of the algorithm, specifically, by the choices made in lines 3, 6, 9 10, 15 and the fact that we have proven this for the sequence returned by \textsc{SubsystemCompiler} in Lemma~\ref{appendix:lemma:alg}.\\

    \textbf{Statement 1.4:} \emph{The returned sequence $G_1,...,G_L$ is such that $ad_{G_1}\cdots ad_{G_{L-1}}(G_L)\propto iV\otimes W$.} This is true by construction of the algorithm, specifically, it is ensured by the sequences collected in lines 3, 7, 8, 9, 14 and 15.\\

    \textbf{Statement 2:} \emph{The algorithm finishes.} Lemma~\ref{lemma:alg} shows that \textsc{SubsystemCompiler} finishes and returns a finite sequence. Since all choices are valid and return finite sequences, the algorithm must finish.
\end{proof}
The proof of the theorem follows by noting that, since the algorithm finishes for any input  $P=iV\otimes W\in i\mathcal{P}^*_{N}$, it must be that $\mathcal{A}'\cup\mathcal{B}'$ is adjoint universal and therefore, $\textrm{span}_{\mathbb{R}}\langle  \mathcal{A}'\cup\mathcal{B}'\rangle_{[\cdot,\cdot]}=\mathfrak{su}(2^N)$.

\section{The SubsystemCompiler algorithm}\label{appendix:lemma:alg}

\begin{figure}[ht!]
\begin{minipage}{\linewidth} 
\begin{algorithm}[H]
\caption{\textsc{SubsystemCompiler}($W$)}
\begin{algorithmic}[1]
\Require Target: $W\in i\mathcal{P}^*_{N-k}$. $\mathcal{A},\mathcal{A}',\mathcal{B},\mathcal{B}'$ as in Theorem~\ref{thm:app:construction}. \vspace{0.2cm}
\State $\textrm{Choose } U_1\otimes B_1,...,U_r\otimes B_r\in\mathcal{B}' \textrm{ s.t. } \prod_{i=1,...,r}B_i\propto W \wedge r=\mathcal{O}(N)$\Comment{Use product universality of $\mathcal{B}$.}
\State $i \gets r-1$
\State $\mathcal{G} = (U_r\otimes B_r,\: )$
\State $\mathcal{H} = \emptyset$
\While{$i\geq 1$}
\If{$\prod_{j\geq i}U_j\prod_{A\in\mathcal{H}}A\propto I$} \Comment{Use ``helper'' Paulis to ensure that no identities appear in the product.}
\State \begin{varwidth}[t]{\linewidth}
$\textrm{Choose $A_1\not\propto A_2\in\mathcal{A}$} \textrm{ s.t. } [A_1,U_i]\neq 0$ \par 
\hspace{3.31cm} $\wedge\: [A_2,U_i]\neq 0$ \par
\hspace{3.48cm} $\wedge\: [A_1,A_2]=0.$
\end{varwidth} \Comment{Define helper Paulis.}
\State $\mathcal{H}\gets (A_1,A_2)$
\State $\mathcal{G}\textrm{.append}(A_1\otimes I,A_2\otimes I)$ 
    \Comment{Add helper Paulis to sequence.}
\ElsIf{$\left[ U_{i}\otimes B_{i}, \prod_{j>i} U_j\otimes B_j \prod_{A\in\mathcal{H}} A\otimes I \right]=0$}
    \Comment{Use helper Pauli to make commuting elements anticommute.}
\State \begin{varwidth}[t]{\linewidth}
$\textrm{Choose $A'\in\mathcal{A}$} \textrm{ s.t. } [A',\prod_{j>i}U_j\prod_{A\in\mathcal{H}}A]\neq 0$ \par
\hspace{0.5cm}$\wedge\: [A',U_{i}]\neq 0.$ \par
\hspace{2.13cm}$\wedge\: A'\not\propto \prod_{j\geq i}U_j\prod_{A\in\mathcal{H}}A.$
\end{varwidth} 
    \Comment{Define helper Pauli.}
\State $\mathcal{H}\gets A'$
\State $\mathcal{G}.\textrm{append}(A'\otimes I)$ 
    \Comment{Add helper Pauli to sequence.}
\Else
\State $\mathcal{G}.\textrm{append}(U_i\otimes B_i)$ 
    \Comment{Add next element of the product to the sequence.}
\State $i\gets i-1$
\EndIf
\EndWhile\\
\Return $\mathcal{G}=(G_1,...,G_{|\mathcal{G}|})$ 
    \Comment{$[G_{|\mathcal{G}|},[ \cdots,[G_2,G_1]\cdots]]=
\alpha V'\otimes W$ for $V'\in\mathcal{P}^*_k$ and $\alpha\in\mathbb{R}^*$}
\end{algorithmic}
\end{algorithm}
% \vspace{-1.75\baselineskip}
\end{minipage}
\caption{\textsc{SubsystemCompiler} with line numbering and comments for reference in the proof of correctness.}
\label{fig:app-alg}
\end{figure}

The proof of \Cref{thm:app:construction} is established via the algorithm \textsc{PauliCompiler}, which in turn calls the algorithm \textsc{SubsystemCompiler} as a subroutine. Here we prove the correctness of \textsc{SubsystemCompiler} through the following Lemma:
    \begin{Lemma}\label{lemma:alg}
        For any input  $W\in i\mathcal{P}^*_{N-k}$ and sets $\mathcal{A}',\mathcal{B}'$ as in Theorem~\ref{thm:app:construction}, \textsc{SubsystemCompiler} finishes and returns a finite, ordered set of operators $\mathcal{G}=(G_1,...,G_{|\mathcal{G}|})$ with $G_i\in \mathcal{A}'\cup \mathcal{B}'$ for all $i=1,...,|\mathcal{G}|$ such that $\exists\alpha\in\mathbb{R}^*$ with $ad_{G_{|\mathcal{G}|}}\cdots ad_{G_2}(G_1)= \alpha V'\otimes W$ for some $V'\in \mathcal{P}^*_k$.
    \end{Lemma}
    \begin{proof}
        The algorithm is depicted in Fig.~\ref{fig:app-alg} with line numbering that will be referred to in this proof. \\

        \textbf{Statement 1:} \emph{The algorithm is correct.}
        Let us start by showing that the algorithm is correct, i.e., it returns a sequence of operators $G_1,...,G_{|\mathcal{G}|}$ such that $ad_{G_{|\mathcal{G}|}}\cdots ad_{G_2}(G_1)\propto V'\otimes W$ for some $V'\in \mathcal{P}^*_k$ is nonzero.
        Given a target Pauli $W\in i\mathcal{P}_{N-k}$, in line 1, we first find a sequence of operators $U_1\otimes B_1,...,U_r\otimes B_r\in\mathcal{B}'$ such that $\prod_{i=1,...,r}B_i\propto W$. This is always guaranteed to exist since $\mathcal{B}$ is product universal. As we discuss in Appendix~\ref{appendix:thm:optimal}, the optimal choice is $r=\mathcal{O}(N)$. However, it might be that $ad_{U_1\otimes B_1}\cdots ad_{U_{r-1}\otimes B_{r-1}}(U_r\otimes B_r)=0$. Therefore, the algorithm goes iteratively through this sequence, starting with $r$ and adds operators $A\otimes I\in\mathcal{A}'$ between terms of the product to ensure that the nested commutator never gives zero. Let us prove the correctness of these insertions through four further statements:\\

        \textbf{Statement 1.1:} \emph{All choices in the algorithm are valid.} There are three choices being made in the algorithm, in line 1, 7 and 11. We already proved the existence of the choice in line 1 with $r=\mathcal{O}(N)$.
        
        Line 7: For simplicity, we identify $A\equiv A_1$, $B\equiv A_2$ and $C\equiv U_i$. Let $C\in \mathcal{P}^*_k$ be arbitrary. ($U_i\neq I$ by construction of $\mathcal{B}'$.) Then, we show that there exist $A,B\in \mathcal{P}^*_k$ such that $A\neq B$ and $[A,C]\neq 0$ and $[B,C]\neq 0$ and $[A,B]=0$ iff $k>1$. This is not true for $k=1$ and  true for $k=2$. The rest follows by induction. 
        To see this, consider $C_1\otimes C_2\in\mathcal{P}^*_{k+1}$ for $C_1\in \mathcal{P}_k$ and $C_2\in\mathcal{P}_1$. If $C_1\neq I$ there exist such $A,B$ by assumption. If $C_1=I$, then choose $A=X\otimes X \otimes A_2$ and $B=Z\otimes Z \otimes A_2$ with $A_2\in\mathcal{P}_1$ such that $[A_2,C]\neq 0$ (where we assumed an identity on all but the first two and last qubits).

        Line 11: 
        For simplicity, we identify $A\equiv A'$, $B\equiv \prod_{j>i}U_j\prod_{A\in\mathcal{H}}A$ and $C\equiv U_i$. Let $B\neq C\in\mathcal{P}^*_k$ be arbitrary. (As we will see below, for any $i$, if we get to line 11, then $B\neq I$.) Then, we show that there exist $A\in\mathcal{P}^*_k$ such that $[A,B]\neq 0$ and $[A,C]\neq 0$ and $A\not\propto BC$ iff $k>1$. 
        Note that for any two Pauli operators $B\neq C\in\mathcal{P}^*_k$, we can find another Pauli $A\in\mathcal{P}_k$ that anticommutes with both. However, for $k=1$ this Pauli must be the product of the other two. Consider $k>1$ for two commuting Pauli operators $B\neq C$. Since then, $[BC,B]=0$, any such anticommuting $A$ cannot be equal to $BC$.  Now, consider two mutually anticommuting Pauli operators $B\neq C$. Choose $D=BC$ which anticommutes with both $B$ and $C$. Since $k>1$, there exist a Pauli operator $E\in\mathcal{P}_k^*$ that commutes with both $B$ and $C$, allowing us to choose $A=DE$. 
        
        Let us now see that the requirement $B\neq I$ is given, i.e., $\prod_{j> i}U_j\prod_{A\in\mathcal{H}}A\neq I$ at line 11 for any iteration $i$. Clearly, this is true for the first iteration $i=r-1$. Now note that at iteration $i$ of the algorithm, the next step can only be reached through line 16 where $i$ becomes $i-1$. This line can only be accessed once lines 6 and 10 evaluate to false. This means that at line 16, $\prod_{j\geq i}U_j\prod_{A\in\mathcal{H}}A\neq I$. Going to the next iteration $i\leftarrow i-1$, at line 6 the previous inequality reads $\prod_{j> i}U_j\prod_{A\in\mathcal{H}}A\neq I$. If line 6 evaluates to false, this is also true at line 11. If line 6 evaluates to true, then the constraints on $A_1,A_2$ are such that $A_1A_2\neq U_i$ and therefore $\prod_{j> i}U_j\prod_{A\in\mathcal{H}\cup \{A_1,A_2\}}A\neq I$. This can be seen as follows: Assume $A_1A_2=U_i$. Then, $[A_1,A_2]=0$ and therefore, $[A_1,A_1A_2]=0$. However, $[A_1,A_1A_2]=[A_1,U_1]\neq 0$, which is a contradiction. That is, we can assume $B\neq I$ above.\\

        \textbf{Statement 1.2:} \emph{$G_1,...,G_{|\mathcal{G}|}\in\mathcal{A}'\cup\mathcal{B}'$.} This is true by construction of the algorithm as the elements appended to the returned sequence $\mathcal{G}$ in lines 3, 9, 13 and 15 have been selected from $\mathcal{A}'$ or $\mathcal{B}'$ in lines 1, 7 and 11. \\
        
        \textbf{Statement 1.3:} $ad_{G_{|\mathcal{G}|}}\cdots ad_{G_2}(G_1)\neq 0$. We must show that for any $l=1,...,|\mathcal{G}|$, $[G_l, \prod_{j<l}G_j]\neq 0$. There are two cases: $G_l\in \mathcal{B}'$ and $G_l\in \mathcal{A}'$. 
        
        Let $G_l\in \mathcal{B}'$. This means that the algorithm has reached line 15 for some iteration. Therefore, we know that line 10 evaluated as false for some $G_l=U_i\otimes B_i$. At that iteration $i$ of the algorithm, $\prod_{j>i}U_j\otimes B_j\prod_{A\in\mathcal{H}}A\otimes I=\alpha\prod_{j<l}G_l$ for some $\alpha\in\mathbb{R}^*$.  Therefore, $[G_l, \prod_{j<l}G_l]\neq 0$.

        Let $G_l\in \mathcal{A}'$. This operator has been added at some iteration $i$ in line 7 or 11 of the algorithm. Assume first that $G_l$ has been added at line 7 at iteration $i$. At that iteration, we have $\prod_{j<l}G_j=\alpha \prod_{j>i}U_j\otimes B_j\prod_{A\in\mathcal{H}}A\otimes I$ for some $\alpha\in\mathbb{R}^*$. At line 7, we also have $U_i\propto \prod_{j>i}U_j\prod_{A\in\mathcal{H}}A$. $A_1,A_2$ are both chosen such that they mututally commute, but both anticommute with $U_i$ and hence, in both cases, $[G_l,\prod_{j<l}G_j]\neq 0$. Now assume $G_l$ has been added at line 11 at iteration $i$. At that iteration, we have $\prod_{j<l}G_j=\alpha \prod_{j>i}U_j\otimes B_j\prod_{A\in\mathcal{H}}A\otimes I$ for some $\alpha\in\mathbb{R}^*$. $A'$ is chosen such that it anticommutes with $\prod_{j>i}U_j\otimes B_j\prod_{A\in\mathcal{H}}A\otimes I$ and therefore, $[G_l,\prod_{j<l}G_j]\neq 0$.\\

        \textbf{Statement 1.4:}  $ad_{G_{|\mathcal{G}|}}\cdots ad_{G_2}(G_1)\propto V'\otimes W$. Since $ad_{U_1\otimes B_1}\cdots ad_{U_{r-1}\otimes B_{r-1}}(U_r\otimes B_r)\propto V''\otimes W$ and the algorithm only adds operators $A\otimes I\in\mathcal{A}'$, it follows that $ad_{G_{|\mathcal{G}|}}\cdots ad_{G_2}(G_1)\propto V'\otimes W$.\\

        \textbf{Statement 2:} \emph{The algorithm finishes.}
        Let us conclude by showing that the algorithm finishes, i.e., it returns the output in finite steps.
        The algorithm would not finish if the while loop always evaluates as true, i.e. $i$ is never reduced to 0. $i$ is reduced only in line 16, so we have to make sure that for any $i$, both line 6 and 10 will eventually evaluate as false. Let us show this by proving the following statements:\\
        
        \textbf{Statement 2.1:} \emph{If line 6 evaluates as true at some iteration, adding $A_1,A_2$ (as chosen in line 7) to $\mathcal{H}$ will make it evaluate as false in the following loop.} If at some iteration $i$, $\prod_{j\geq i}U_j\prod_{A\in\mathcal{H}}A\propto I$, then $A_1A_2\prod_{j\geq i}U_j\prod_{A\in\mathcal{H}}A\propto A_1A_2$. Since $A_1\neq A_2$ as by the choice in line 7, the statement is true.\\
        
        \textbf{Statement 2.2:} \emph{If line 10 evaluates as true at some iteration, adding $A'$ (as chosen in line 11) to $\mathcal{H}$ will make it evaluate as false in the following loop.} This is due to the first constraint in line 11. If at some iteration $i$, $\left[ U_{i}\otimes B_{i}, \prod_{j>i} U_j\otimes B_j \prod_{A\in\mathcal{H}} A\otimes I \right]=0$, then, $\left[ U_{i}\otimes B_{i}, A'\otimes I\prod_{j>i} U_j\otimes B_j \prod_{A\in\mathcal{H}} A\otimes I \right]\neq 0$ because $[U_i\otimes B_i,A_i\otimes I]\neq 0$. Hence, the statement is true.\\
        
        \textbf{Statement 2.3:} \emph{If line 10 evaluates as true at some iteration, adding $A'$ (as chosen in line 11) to $\mathcal{H}$ will make line 6 evaluate as false in the following loop.} This is due to the third constraint in line 11. Since $A'\not\propto \prod_{j\geq i}U_j\prod_{A\in\mathcal{H}}A$ as by the choice in line 11, $A'\prod_{j\geq i}U_j\prod_{A\in\mathcal{H}}A\not\propto I$ and therefore, line 6 evaluates as false.\\
        
        Since all choices are provably valid and return finite sequences, line 16 will deterministically be reached and the algorithm finishes. This also implies that the returned sequence is finite.
    \end{proof}

\section{Proof of Proposition~1} \label{app:proof_prop_su4}

Recall the statement of Proposition~1 from the main text:
\begin{Proposition}  For $\mathcal{A} = \{iX_{1}, iZ_{1}, iX_{2}, iZ_{2}, iZ_{1} \otimes Z_{2}\}$, we have that
\begin{align}
\langle \mathcal{A} \rangle_{\textrm{Lie}} = \mathfrak{su}(4).
\end{align}
\end{Proposition}

\begin{proof} The proof proceeds by direct calculation. Note that we have
\begin{align}
i\mathcal{P}_{2}^{*} = \{iA \otimes B | A,B \in \{I,X,Y,Z \} \} \setminus \{iI \otimes I \}.
\end{align}
Since $[iX, iZ] = -2iY$, $\langle \mathcal{A} \rangle_{[\cdot, \cdot]}$ clearly contains all elements of $i\mathcal{P}_{2}^{*}$ with a single non-identity term. By noting that
\begin{gather}
[iX_{1}, iZ_{1}\otimes Z_{2}] = -2iY_{1} \otimes Z_{2}, \\
[iZ_{1},[iX_{1}, iZ_{1} \otimes Z_{2}] = 4i X_{1} \otimes Z_{2},
\end{gather}
and similarly for taking commutators with $X_{2}$ and $Z_{2}$, we see that $\langle \mathcal{A} \rangle_{[\cdot, \cdot]}$ contains all Pauli strings with two non-identity terms and hence $\langle \mathcal{A} \rangle_{[\cdot, \cdot]} = i\mathcal{P}_{2}^{*}$, proving the claim.
\end{proof}

\section{The Pauli Group and $\mathbb{F}_{2}^{2N}$} \label{app:symplectic_explanation}

In this appendix, we explain the connection between the Pauli group on $N$ qubits and the symplectic space $\mathbb{F}_{2}^{2N}$ ($\mathbb{F}_{2}$ denotes the finite field containing two elements) alluded to in the main text and used in a number of proofs below. 

Let us denote by $\mathcal{G}_{\mathcal{P}_{N}}$ the Pauli group on $N$ qubits defined via
\begin{align}
\mathcal{G}_{\mathcal{P}_{N}} := \left\{ e^{i\theta} P | \theta = 0,\frac{\pi}{2}, \pi, \frac{3\pi}{2}; P \in \mathcal{P}_{N} \right\}
\end{align}
with the group multiplication given by standard matrix multiplication. As in the main text, $\mathcal{P}_{N}$ denotes the set of Pauli strings, i.e.
\begin{align}
\mathcal{P}_{N} := \big\{P = p_{1} \otimes p_{2} \otimes ... \otimes p_{N} | p_{i} \in \{I, X, Y, Z\}, \text{ for each } i = 1, ..., N \big\}.
\end{align}
Since the single qubit Pauli operators satisfy $Y = iXZ$, it is possible to write any element $e^{i\theta} P \in \mathcal{G}_{\mathcal{P}_{N}}$ as $e^{i\theta'}P_{X}P_{Z}$ where $P_{X} \in \mathcal{P}_{N}$ is the Pauli string consisting of $X$ terms at every tensor factor where $P$ contained either an $X$ or a $Y$ and identities elsewhere, $P_{Z} \in \mathcal{P}_{N}$ is similarly defined with $Z$ terms at every factor where $P$ contained either a $Z$ or $Y$ and identities elsewhere, and $\theta' = \theta + y \frac{\pi}{2}$ where $y$ is the number of $Y$ terms contained in $P$. We say that elements of $\mathcal{G}_{\mathcal{P}_{N}}$ written in this way are in \textit{standard form}.

Let $e^{i\theta_{P}}P_{X}P_{Z}$ and $e^{i\theta_{Q}}Q_{X}Q_{Z}$ be elements of $\mathcal{G}_{\mathcal{P}_{N}}$ in standard form. We see that the standard from of their product can be written as
\begin{align}
\left(e^{i\theta_{P}}P_{X}P_{Z} \right)\left( e^{i\theta_{Q}}Q_{X}Q_{Z}\right) = e^{i(\theta_{P} + \theta_{Q} + \lambda_{PQ}\pi)}\left(P_{X}Q_{X}\right)\left(P_{Z}Q_{Z}\right)
\end{align}
where $\lambda_{PQ} \in \{0,1\}$ is the parity of the number of tensor factors where $P_{Z}$ and $Q_{X}$ both contain a non-identity term. Accordingly, the term $\lambda_{PQ}$ contains all the information about whether $e^{i\theta_{P}}P_{X}P_{Z}$ and $e^{i\theta_{Q}}Q_{X}Q_{Z}$ commute or not: $\lambda_{PQ} = 0$ if they commute and $\lambda_{PQ} = 1$ otherwise.

Clearly, $\mathcal{P}_{N}$ is a subset of $\mathcal{G}_{\mathcal{P}_{N}}$. When $\mathcal{S} \subset \mathcal{P}_{N}$ is such that (i) $\mathcal{S}$ is a subgroup of $\mathcal{G}_{\mathcal{P}_{N}}$, (ii) $-I \notin \mathcal{S}$ and (iii) all elements of $\mathcal{S}$ commute, then $\mathcal{S}$ is called a stabiliser subgroup \cite{nielsen2010quantum} and is used throughout many areas of quantum computing. Many results arising in e.g., quantum error correction make use of the number of independent generators of the group $\mathcal{S}$, that is, the smallest number of elements of $\mathcal{S}$ that generate $\mathcal{S}$ and are not products of each other. A convenient way to check independence is via the map $f:\mathcal{S} \rightarrow \mathbb{F}_{2}^{2N}$ defined via
\begin{align}
e^{i\theta}P_{X}P_{Z} \mapsto \boldsymbol{v}_{P}
\end{align}
where $e^{i\theta}P_{X}P_{Z}$ is the standard form of an element of $\mathcal{S}$ and $\boldsymbol{v}_{P}$ is the vector containing a $1$ at entry $j$ if the $j$th tensor factor of $P_{X}$ contains an $X$, a $1$ at entry $N+j$ if the $j$th tensor factor of $P_{Z}$ contains a $Z$, and zeros otherwise. It follows that the above map is a group isomorphism where $\mathbb{F}_{2}^{2N}$ is considered as an additive group, meaning that the product of elements in $\mathcal{S}$ corresponds to the sum of elements in $\mathbb{F}_{2}^{2N}$. In fact, it is convenient to take $\mathbb{F}_{2}^{2N}$ to be a vector space over the finite field $\mathbb{F}_{2}$. By doing so, the independence of elements of $\mathcal{S}$ can be rephrased as linear independence of the corresponding elements in $\mathbb{F}_{2}^{2N}$. 

Let us now consider the map $f$ extended to all of $\mathcal{G}_{\mathcal{P}_{N}}$. This map remains a a group homomorphism however it is now many-to-one since it sends every element of $\mathcal{G}_{\mathcal{P}_{N}}$ with standard forms that differ only in the angle $\theta$ to the same element of $\mathbb{F}_{2}^{2N}$. Nevertheless, when restricted to $\mathcal{P}_{N} \subset \mathcal{G}_{\mathcal{P}_{N}}$ or $i\mathcal{P}_{N}^{*}$, this map remains useful for our purposes as many of the results presented in the main text make statements regarding equivalence up to proportionality. We make this intuition more precise below.

Before doing so, let us add one last piece of structure to $\mathbb{F}_{2}^{2N}$. Define the matrix $\Lambda \in M_{2N}(\mathbb{F}_{2})$ to be 
\begin{align}
\Lambda = \begin{bmatrix} \boldsymbol{0} & I_{N} \\ I_{N} &\boldsymbol{0}
\end{bmatrix}
\end{align}
where $\boldsymbol{0} \in M_{N}(\mathbb{F}_{2})$ is the matrix of all zeros and $I_{N} \in M_{N}(\mathbb{F}_{2})$ is the identity matrix. When equipped with the symplectic bilinear form $\Lambda: \mathbb{F}_{2}^{2N} \times \mathbb{F}_{2}^{2N} \rightarrow \mathbb{F}_{2}$ defined via $(\boldsymbol{v}, \boldsymbol{w}) \mapsto \boldsymbol{v}^{\top}\Lambda \boldsymbol{w}$, the vector space $\mathbb{F}_{2}^{2N}$ becomes a symplectic vector space. This symplectic form captures the commutativity of elements of $\mathcal{G}_{\mathcal{P}_{N}}$ exactly: $P, Q \in \mathcal{G}_{\mathcal{P}_{N}}$ commute if and only if $\boldsymbol{v} := f(P)$ and $\boldsymbol{w} := f(Q)$ are such that $\boldsymbol{v}^{\top}\Lambda \boldsymbol{w} = 0$. 

The following lemma establishes the justification for being able to work with $\mathbb{F}_{2}^{2N}$ for proving certain statements made in the main text. For any subset $\mathcal{C} \subset \mathcal{G}_{\mathcal{P}_{N}}$, let $f(\mathcal{C}) \subset \mathbb{F}_{2}^{2N}$ denote the image of $\mathcal{C}$. Note that, this in particular means that $f(i\mathcal{P}_{N}^{*}) = \mathbb{F}_{2}^{2N}\setminus \{\boldsymbol{0}\}$ where $\boldsymbol{0} \in \mathbb{F}_{2}^{2N}$ is the vector containing only $0$ entries. For $\mathcal{A} \subset i \mathcal{P}_{N}^{*}$, let us define
\begin{align}
\langle \mathcal{A} \rangle_{\| \cdot \|_{2}} := \left\{ \frac{A}{\|A\|_{2}} | A \in \langle \mathcal{A} \rangle_{[\cdot, \cdot]} \right\}
\end{align}
where $\|\cdot\|_{2}$ denotes the spectral norm (in fact other choices of norm also work). Note that $\langle \mathcal{A} \rangle_{\| \cdot \|_{2}} \subset i\mathcal{P}_{N}^{*}$. We have:
\begin{Lemma} \label{lem:sympelctic_justification} Let $\mathcal{A}, \mathcal{B} \subset i \mathcal{P}_{N}^{*}$. Then
\begin{enumerate}[label=(\roman*)]
    \item $\mathcal{A}$ is adjoint universal for $i\mathcal{P}_{N}^{*}$ if and only if $f(\langle \mathcal{A} \rangle_{\| \cdot \|_{2}})$ contains $\mathbb{F}_{2}^{2N}\setminus \{\boldsymbol{0}\}$.
    \item $\mathcal{B}$ is product universal for $i\mathcal{P}_{N}^{*}$ if and only if $f(\mathcal{B})$ contains a basis of $\mathbb{F}_{2}^{2N}$.
\end{enumerate}
\end{Lemma}

\begin{proof} The proof essentially follows from the observation that there is a one-to-one correspondence between $i\mathcal{P}_{N}^{*}$ and $f(i\mathcal{P}_{N}^{*}) = \mathbb{F}_{2}^{2N} \setminus \{\boldsymbol{0}\}$, and hence also between any subset $\mathcal{C} \subset i\mathcal{P}_{N}^{*}$ and $f(\mathcal{C})$. 

\textbf{Case (i):} Suppose that $\mathcal{A}$ is adjoint universal for $i\mathcal{P}_{N}^{*}$. This means that for each $A \in i\mathcal{P}_{N}^{*}$ there is some non-zero $\alpha \in \mathbb{R}$ such that $\alpha A \in \langle \mathcal{A} \rangle_{[\cdot, \cdot]}$. In particular, this means that $\frac{\alpha}{|\alpha|}A \in \langle \mathcal{A} \rangle_{\| \cdot \|_{2}}$ and hence that $f(A) \in f(\langle \mathcal{A} \rangle_{\| \cdot \|_{2}})$. It follows that $f(i \mathcal{P}_{N}^{*}) \subseteq f(\langle \mathcal{A} \rangle_{\| \cdot \|_{2}})$.

Suppose now that $f(\langle \mathcal{A} \rangle_{\| \cdot \|_{2}})$ contains $\mathbb{F}_{2}^{2N} \setminus \{\boldsymbol{0}\}$. In particular, this means that $f(\langle \mathcal{A} \rangle_{\| \cdot \|_{2}}) = \mathbb{F}_{2}^{2N} \setminus \{\boldsymbol{0}\}$. Since each $\boldsymbol{v} \in \mathbb{F}_{2}^{2N} \setminus \{\boldsymbol{0}\}$ corresponds to some $A \in i\mathcal{P}_{N}^{*}$, it follows that there is some element $A' \in \langle \mathcal{A} \rangle_{\| \cdot \|_{2}}$ such that $A = e^{i\theta}A'$ for some $\theta \in \{0, \frac{\pi}{2}, \pi, \frac{3\pi}{2}\}$. However, since both $A$ and $A'$ are in $i \mathcal{P}_{N}^{*}$, it must be the case that $A = \pm A'$. Since $A' \in \langle \mathcal{A} \rangle_{\| \cdot \|_{2}}$ means that there is some $A'' \in \langle \mathcal{A}\rangle_{[\cdot, \cdot]}$ such that $A''/\|A''\|_{2} = A'$, it follows that $\alpha A \in \langle \mathcal{A}\rangle_{[\cdot, \cdot]}$ for $\alpha := \|A''\|_{2}$.

\textbf{Case (ii):} Suppose that $\mathcal{B}$ is product universal for $i\mathcal{P}_{N}^{*}$. Recall that this means that for any $A \in i\mathcal{P}_{N}^{*}$, there exists a sequence $B_{i_{1}},B_{i_{2}},...,B_{i_{r}} \in \mathcal{B}$ such that $A = e^{i\theta}B_{i_{1}}...B_{i_{r}}$ for some $\theta \in \{0, \frac{\pi}{2}, \pi, \frac{3\pi}{2}\}$. In particular, this means that $f(A) = f(B_{i_{1}}...B_{i_{r}}) = f(B_{i_{1}}) + ... + f(B_{i_{r}})$. Noting that $f(A) \neq \boldsymbol{0}$ since $A \in i \mathcal{P}_{N}^{*}$ and hence also that $f(B_{i_{1}}) + ... + f(B_{i_{r}}) \neq \boldsymbol{0}$. The latter implies that some subset of $\{f(B_{i_{1}}), ..., f(B_{i_{r}})\}$ is linearly independent. The union of all such subsets for each $A \in i\mathcal{P}_{N}^{*}$ necessarily contains a basis of $\mathbb{F}_{2}^{2N}$ since $f(i\mathcal{P}_{N}^{*}) = \mathbb{F}_{2}^{2N} \setminus \{\boldsymbol{0}\}$.

Suppose that $f(\mathcal{B})$ contains a basis of $\mathbb{F}_{2}^{2N}$. Let this basis be denoted by $\{\boldsymbol{v}_{1}, ..., \boldsymbol{v}_{2N}\}$. Since there is a one-to-one correspondence between $\mathcal{B}$ and $f(\mathcal{B})$, we can associate each $\boldsymbol{v}_{j}$ with some $B_{j} \in \mathcal{B}$. Since any element $\boldsymbol{w} \in \mathbb{F}_{2}^{2N} \setminus \{\boldsymbol{0}\}$ an be written as $\boldsymbol{w} = \boldsymbol{v}_{i_{1}} + ... + \boldsymbol{v}_{i_{r}} = f(B_{i_{1}}) + ... + f(B_{i_{r}}) = f(B_{i_{1}}...B_{i_{r}})$ and since $\boldsymbol{w} = f(A)$ for some $A \in i\mathcal{P}_{N}^{*}$, it follows that $A = e^{i\theta}B_{i_{1}}...B_{i_{r}}$ for some $\theta \in \{0, \frac{\pi}{2},\pi, \frac{3\pi}{4} \}$.
\end{proof}

Having demonstrated that it suffices to work entirely with $\mathbb{F}_{2}^{2N}$ for the purposes of this paper, let us introduce the following notation for later use. In direct parallel to the maps $ad_{A}$ in the main text, let us define the map $ad_{\boldsymbol{v}}: \mathbb{F}_{2}^{2N} \rightarrow \mathbb{F}_{2}^{2N}$ via $ad_{\boldsymbol{v}}(\boldsymbol{w}) = \left(\boldsymbol{v}^{\top}\Lambda \boldsymbol{w} \right)(\boldsymbol{v} + \boldsymbol{w})$. We note that, for all $A, B \in i \mathcal{P}_{N}^{*}$ we have that
\begin{align}
ad_{A}(B) = \begin{cases} \pm 2AB, &\text{ if } A,B \text{ anticommute }, \\ 0 &\text{ otherwise}
\end{cases}
\end{align}
and
\begin{align}
ad_{f(A)}(f(B)) = \begin{cases} f(A) + f(B),  &\text{ if } A,B \text{ anticommute }, \\ 0 &\text{ otherwise}.
\end{cases}
\end{align}
That is, $ad_{A}(B) \neq 0$ if and only if $ad_{f(A)}(f(B)) \neq 0$, and moreover if $ad_{A}(B) \neq 0$, then 
\begin{align}
f\left( \frac{ad_{A}(B)}{\|ad_{A}(B) \|_{2}} \right) = ad_{f(A)}(f(B)).
\end{align}
Since, for any $\mathcal{A} \subset i\mathcal{P}_{N}^{*}$,
\begin{align}
\textrm{span}_{\mathbb{R}}\langle \mathcal{A} \rangle_{[\cdot, \cdot]} = \textrm{span}_{\mathbb{R}}\langle \mathcal{A} \rangle_{\|\cdot \|_{2}},
\end{align}
we will typically make no distinction between results adjoint universality proven in the context of $i\mathcal{P}_{N}^{*}$ and those proven in $\mathbb{F}_{2}^{2N}$, and drop reference to $\|\cdot\|_{2}$ from future discourse.

\section{Proof of Theorem~2} \label{app:proof_thm_2n}

Recall the statement of Theorem~2 from the main text:
\begin{Theorem} Let $\mathcal{A}\subset i\mathcal{P}_{N}^{*}$ consist of $2N$ elements. Then, 
\begin{align*}
\langle \mathcal{A}\rangle_{\mathrm{Lie}}\neq \mathfrak{su}(2^N).
\end{align*}
\end{Theorem}
To prove this theorem, we make use of the relation between $i\mathcal{P}_{N}^{*}$ and $\mathbb{F}_{2}^{2N}$ as established above in \Cref{app:symplectic_explanation}. To do so, we first make some comments regarding the assumptions on $\mathcal{A}$, introduce some notation and state an equivalent formulation of the theorem in the context of $\mathbb{F}_{2}^{2N}$.

It is sufficient to take $\mathcal{A}$ to be a set consisting of $2N$ independent elements since if this were not the case then there exists an element of $i\mathcal{P}_{N}^{*}$ independent of all elements of $\mathcal{A}$ which is thus necessarily not an element of $\langle \mathcal{A} \rangle_{[\cdot, \cdot]}$. In particular, this means that $f(\mathcal{A})$ is a basis for $\mathbb{F}_{2}^{2N}$. In the following, we will denote a basis for $\mathbb{F}_{2}^{2N}$ by $\widehat{\mathcal{A}}$.

Let use define, for any $\widehat{\mathcal{A}}$, the following analogue of $\langle \mathcal{A} \rangle_{[\cdot, \cdot]}$ in the main text:
\begin{align}
\langle \widehat{\mathcal{A}} \rangle_{[\cdot, \cdot]} = \widehat{\mathcal{A}} \bigcup_{r = 1}^{\infty} \{ad_{\boldsymbol{v}_{i_{1}}} \cdots ad_{\boldsymbol{v}_{i_{r}}}(\boldsymbol{v})| (\boldsymbol{v}_{i_{1}}, ..., \boldsymbol{v}_{i_{r}}, \boldsymbol{v}) \in \widehat{\mathcal{A}}^{r+1} \}.
\end{align}
Accordingly, proving Theorem~2 amounts to proving the following equivalent theorem:
\begin{Theorem} \label{thm:2n_rephrased} Let $\widehat{\mathcal{A}}$ be a basis of $\mathbb{F}_{2}^{2N}$ where $N \geq 3$. Then
\begin{align}
\langle \widehat{\mathcal{A}}\rangle_{[\cdot, \cdot]} \neq \mathbb{F}_{2}^{2N}.
\end{align}
\end{Theorem}

\begin{proof} The proof proceeds by demonstrating that for any choice of $\widehat{\mathcal{A}}$, there always exists an element of $\mathbb{F}_{2}^{2N}$ that is not contained in $\langle \widehat{\mathcal{A}}\rangle_{[\cdot, \cdot]}$. We split this into two cases: (i) where $\widehat{\mathcal{A}}$ every element $\boldsymbol{v} \in \widehat{\mathcal{A}}$ is such that $\boldsymbol{v}^{\top}\Lambda \boldsymbol{w} = 1$ for all $\boldsymbol{w} \in \widehat{\mathcal{A}}$ such that $\boldsymbol{w} \neq \boldsymbol{v}$, and (ii) where there exist distinct elements $\boldsymbol{v}$ and $\boldsymbol{w}$ in $\widehat{\mathcal{A}}$ such that $\boldsymbol{v}^{\top}\Lambda \boldsymbol{w} = 0$. Case (i) is proven as \Cref{prop:all_anti_commuting} below while case (ii) is proven by the conjunction of \Cref{prop:odd_sequence} and \Cref{prop:even_sequences}.
\end{proof}

\begin{Proposition} \label{prop:all_anti_commuting} Let $\widehat{\mathcal{A}}$ be a basis for $\mathbb{F}_{2}^{2N}$ with $N\geq 3$, such that for all distinct $\boldsymbol{v},\boldsymbol{w} \in \widehat{\mathcal{A}}$, $\boldsymbol{v}^{\top}\Lambda\boldsymbol{w} = 1$. Then for any three distinct $\boldsymbol{v}, \boldsymbol{w}, \boldsymbol{x} \in \widehat{\mathcal{A}}$, $\boldsymbol{v} + \boldsymbol{w} + \boldsymbol{x} \notin \langle \widehat{\mathcal{A}}\rangle_{[\cdot,\cdot]}$.
\end{Proposition}

\begin{proof} For any sequence $\boldsymbol{u}_{1}, ..., \boldsymbol{u}_{r} \in \widehat{\mathcal{A}}$ 
\begin{align}
ad_{\boldsymbol{u}_{1}}...ad_{\boldsymbol{u}_{r-1}}(\boldsymbol{u}_{r}) = \boldsymbol{0}
\end{align}
if $r \geq 3$ since 
\begin{align}
ad_{\boldsymbol{u}_{r-2}}ad_{\boldsymbol{u}_{r-1}}(\boldsymbol{u}_{r}) &= \boldsymbol{u}_{r-2}^{\top} \Lambda (\boldsymbol{u}_{r-1} + \boldsymbol{u}_{r})\left[\boldsymbol{u}_{r-2} + \boldsymbol{u}_{r-1} + \boldsymbol{u}_{r} \right] \\
&= \boldsymbol{u}_{r-2}^{\top} \Lambda \boldsymbol{u}_{r-1}\left[\boldsymbol{u}_{r-2} + \boldsymbol{u}_{r-1} + \boldsymbol{u}_{r} \right] \nonumber \\
& \quad \quad + \boldsymbol{u}_{r-2}^{\top} \Lambda \boldsymbol{u}_{r}\left[\boldsymbol{u}_{r-2} + \boldsymbol{u}_{r-1} + \boldsymbol{u}_{r} \right] \\
&= \boldsymbol{0}
\end{align}
and since $ad_{\boldsymbol{u}_{l}}(\boldsymbol{0}) = \boldsymbol{0}$. The result follows since there is no sequence $\boldsymbol{u}_{1}, \boldsymbol{u}_{2}$ such that $ad_{\boldsymbol{u}_{1}}(\boldsymbol{u}_{2}) = \boldsymbol{v} + \boldsymbol{w} + \boldsymbol{x}$.
\end{proof}

\begin{Proposition} \label{prop:odd_sequence} Let $\widehat{\mathcal{A}}$ be a basis for $\mathbb{F}_{2}^{2N}$ such that there exists distinct elements $\boldsymbol{v}, \boldsymbol{w} \in \widehat{\mathcal{A}}$ satisfying $\boldsymbol{v}^{\top} \Lambda \boldsymbol{w} = 0$. There does not exist a sequence $\boldsymbol{u}_{1}, ..., \boldsymbol{u}_{r} \in \widehat{\mathcal{A}}$ for odd $r$ such that
\begin{align}
ad_{\boldsymbol{u}_{1}}...ad_{\boldsymbol{u}_{r-1}}(\boldsymbol{u}_{r}) = \boldsymbol{v} + \boldsymbol{w}.
\end{align}
\end{Proposition}

\begin{proof} Assume that
\begin{align}
ad_{\boldsymbol{u}_{1}}...ad_{\boldsymbol{u}_{r-1}}(\boldsymbol{u}_{r}) = \boldsymbol{v} + \boldsymbol{w}.    
\end{align}
In order to not violate linear independence, we must have that
\begin{align}
ad_{\boldsymbol{u}_{1}}...ad_{\boldsymbol{u}_{r-1}}(\boldsymbol{u}_{r}) \neq \boldsymbol{0}
\end{align}
and that $\boldsymbol{v}$ and $\boldsymbol{w}$ each appear in the sequence an odd number of times. We also know that, if $ad_{\boldsymbol{u}_{1}}...ad_{\boldsymbol{u}_{r-1}}(\boldsymbol{u}_{r}) \neq \boldsymbol{0}$ then
\begin{align}
ad_{\boldsymbol{u}_{1}}...ad_{\boldsymbol{u}_{r-1}}(\boldsymbol{u}_{r}) = \boldsymbol{u}_{1} + ... + \boldsymbol{u}_{r}.
\end{align}
It follows that
\begin{align}
\boldsymbol{u}_{1} + ... + \boldsymbol{u}_{r} = \boldsymbol{v} + \boldsymbol{w} + \boldsymbol{u}_{l_{1}} + ... + \boldsymbol{u}_{l_{t}}.
\end{align}
If $r$ is odd then so is $t$, meaning that either $\boldsymbol{u}_{l_{1}} + ... + \boldsymbol{u}_{l_{t}} \neq \boldsymbol{0}$, contradicting the assumption that $ad_{\boldsymbol{u}_{1}}...\boldsymbol{u}_{r-1}(\boldsymbol{u}_{r}) = \boldsymbol{v} + \boldsymbol{w}$, or $\boldsymbol{u}_{l_{1}} + ... + \boldsymbol{u}_{l_{t}} = \boldsymbol{0}$, contradicting the assumption of linear independence. 
\end{proof}

\begin{Proposition} \label{prop:even_sequences} Let $\widehat{\mathcal{A}}$ be a basis for $\mathbb{F}_{2}^{2N}$ that contains two distinct elements $\boldsymbol{v}, \boldsymbol{w}$ such that $\boldsymbol{v}^{\top}\Lambda \boldsymbol{w} = 0$. Then there does not exist a sequence $\boldsymbol{u}_{1},...,\boldsymbol{u}_{r} \in \widehat{\mathcal{A}}$ with $r = 2k$, $k \in \mathbb{N}$ such that
\begin{align}
ad_{\boldsymbol{u}_{1}}...ad_{\boldsymbol{u}_{r-1}}(\boldsymbol{u}_{r}) = \boldsymbol{v} + \boldsymbol{w}.
\end{align}
\end{Proposition}

\begin{proof} The case where $r = 2$ is trivial by the assumption that $\boldsymbol{v}^{\top}\Lambda \boldsymbol{w} = 0$. Suppose such a sequence $\boldsymbol{u}_{1},...,\boldsymbol{u}_{r}$ exists for $r = 2k > 2$. Since $\boldsymbol{v} + \boldsymbol{w} \neq \boldsymbol{0}$ by linear independence, this means that
\begin{align}
ad_{\boldsymbol{u}_{1}}...ad_{\boldsymbol{u}_{r-1}}(\boldsymbol{u}_{r}) \neq \boldsymbol{0}
\end{align}
which is equivalent to 
\begin{align}
\prod_{i=1}^{r-1}\left[\boldsymbol{u}_{i}^{\top}\Lambda (\boldsymbol{u}_{i+1} + ... + \boldsymbol{u}_{r}) \right]  = 1.
\end{align}
In particular, this implies that
\begin{align}
\sum_{i=1}^{r-1} \left[\boldsymbol{u}_{i}^{\top}\Lambda (\boldsymbol{u}_{i+1} + ... + \boldsymbol{u}_{r}) \right] = r-1 \bmod 2 = 1 \label{eq:sum_constraints}
\end{align}
since for any product of elements in $\mathbb{F}_{2}$ to be $1$, it must be that every term in the product is also $1$ and there are $r-1$ terms in the product. It is possible to rewrite the left-hand side of the above equation as
\begin{align}
\sum_{i=1}^{r-1}\sum_{j > i} \boldsymbol{u}_{i}^{\top}\Lambda \boldsymbol{u}_{j}. \label{eq:sum_sum_i_lambda_j}
\end{align}
Since $ad_{\boldsymbol{u}_{1}}...ad_{\boldsymbol{u}_{r-1}}(\boldsymbol{u}_{r}) \neq \boldsymbol{0}$, we also have that
\begin{align}
ad_{\boldsymbol{u}_{1}}...ad_{\boldsymbol{u}_{r-1}}(\boldsymbol{u}_{r}) = \boldsymbol{u}_{1} + ... + \boldsymbol{u}_{r}.
\end{align}
For $\boldsymbol{u}_{1} + ... + \boldsymbol{u}_{r}$ to equal $\boldsymbol{v} + \boldsymbol{w}$, it must be the case that $\boldsymbol{v}$ and $\boldsymbol{w}$ appear in the sequence $\boldsymbol{u}_{1},...,\boldsymbol{u}_{r}$ an odd number of times and any other distinct element appears in the sequence an even number of times.

Let us consider \Cref{eq:sum_sum_i_lambda_j} for a sequence of the above type. Suppose without loss of generality that the first $\boldsymbol{v}$ appears before the first $\boldsymbol{w}$ in the sequence and let us write $(l_{(0,1)}, l_{(0,2)})$ for the pair of corresponding indices, i.e. $\boldsymbol{u}_{l_{(0,1)}} = \boldsymbol{v}$ and $\boldsymbol{u}_{l_{(0,2)}} = \boldsymbol{w}$. The remaining $r-2$ elements of the sequence can be paired as follows: let $(l_{(i,1)}, l_{(i,2)})$ for $i = 1, ..., k-1$ be a pair of labels such that $l_{(i,1)}$ is the label of the $i$th unpaired element of the sequence, $l_{(i,2)} > l_{(i,1)}$ has not yet been assigned to another pair, and $\boldsymbol{u}_{l_{(i,1)}} = \boldsymbol{u}_{l_{(i,2)}}$. Let us assume without loss of generality that $l_{(1,1)} < l_{(2,1)} < ... < l_{(k-1,1)}$. By the assumptions on the sequence $\boldsymbol{u}_{1}, ..., \boldsymbol{u}_{r}$, it is possible to assign pairs in this way such that
\begin{align}
\{l_{(i,1)}, l_{(i,2)}: i = 0, ...,k-1 \} = \{1, ..., r\}.
\end{align}
Noting that it must be the case that the label corresponding to $r$ must be $l_{(i,2)}$ for some $i$, let us then define
\begin{align}
L_{1} &:= \{l_{(i,1)} : i = 0, ..., k-1\}, \\
L_{2} &:= \{l_{(i,2)}: i = 0, ..., k-1 \} \setminus \{r\}.
\end{align}
We can then rewrite \Cref{eq:sum_sum_i_lambda_j} as
\begin{align}
\sum_{l_{(i,1)} \in L_{1}} \sum_{l_{(i,1)}< j} \boldsymbol{u}_{l_{(i,1)}}^{\top}\Lambda \boldsymbol{u}_{j} + \sum_{l_{(i,2)} \in L_{2}} \sum_{l_{(i,2)} < j'} \boldsymbol{u}_{l_{(i,2)}}^{\top}\Lambda \boldsymbol{u}_{j'}.
\end{align}
Using that $\boldsymbol{u}_{l_{(i,1)}}^{\top}\Lambda \boldsymbol{u}_{l_{(i,2)}} = 0$, either since $\boldsymbol{u}_{l_{(i,1)}} = \boldsymbol{u}_{l_{(i,2)}}$ for $i = 1, ..., k-1$ or since $\boldsymbol{v}^{\top}\Lambda \boldsymbol{w} = 0$ by assumption, and using that addition is modulo $2$, the above expression can be written as:
\begin{align}
\sum_{\substack{l_{(0,1)} < j, \\ j \neq l_{(0,2)} }} \boldsymbol{v}^{\top}\Lambda \boldsymbol{u}_{j} + (1-\delta_{l_{(0,2),r}})\sum_{l_{(0,2)} < j'} \boldsymbol{w}^{\top}\Lambda \boldsymbol{u}_{j'} + \sum_{\substack{l_{(i,1)} \in L_{1}, \\ i \neq 0}} \sum_{l_{ (i,1)}<j''<l_{(i,2)}} \boldsymbol{u}_{l_{(i,1)}}^{\top} \Lambda\boldsymbol{u}_{j''}.
\end{align}
For every $j > l_{(0,1)}$ there are two possible cases: (i) $j = l_{(i,1)}$ for some $i$ meaning that $l_{(0,1)} < l_{(i,1)} < l_{(i,2)}$ or (ii) $j = l_{(i,2)}$ for some $i$ with $l_{(i,1)} < l_{(0,1)}$. In the first case both $\boldsymbol{v}^{\top}\Lambda\boldsymbol{u}_{l_{(i,1)}}$ and $\boldsymbol{v}^{\top}\Lambda \boldsymbol{u}_{l_{(i,2)}}$ appear in the first sum in the above expression, and since $\boldsymbol{u}_{l_{(i,1)}} = \boldsymbol{u}_{l_{(i,2)}}$, these terms cancel in the modulo two addition. For the second case, $\boldsymbol{v}^{\top}\Lambda \boldsymbol{u}_{l_{(i,2)}}$ appears in the first sum while $\boldsymbol{u}_{l_{(i,1)}}^{\top}\Lambda \boldsymbol{v}$ appears in the third term in the above expression, so again there is cancellation (recall that $\boldsymbol{x}^{\top}\Lambda\boldsymbol{y} = \boldsymbol{y}^{\top}\Lambda\boldsymbol{x}$). Accordingly, the first term above fully cancels out. The exact same reasoning can be used to cancel the remaining terms, as follows. For the second term, either $l_{(0,2)} = r$ meaning that this term is zero anyway, otherwise any $j' > l_{(0,2)}$ also satisfies the same two cases above, and the corresponding terms cancel out. For the remaining term, let us denote by $\widetilde{L}_{1}$ the set of labels $l_{(i,1)}$ that have not been canceled out via the above method. For any $i,i'$ such that $i < i'$ and where $l_{(i,1)}, l_{(i', 1)} \in \widetilde{L}_{1}$, we have the same two cases: either (i) $l_{(i,1)} < l_{(i',1)} < l_{(i',2)} < l_{(i,2)}$ in which case the terms $\boldsymbol{u}_{l_{(i,1)}}^{\top}\Lambda \boldsymbol{u}_{l_{(i',1)}}$ and $\boldsymbol{u}_{l_{(i,1)}}^{\top}\Lambda \boldsymbol{u}_{l_{(i',2)}}$ cancel, or (ii) $l_{(i,1)} < l_{(i',1)}  < l_{(i,2)} < l_{(i',2)}$, in which case $\boldsymbol{u}_{l_{(i,1)}}^{\top} \Lambda \boldsymbol{u}_{l_{(i',1)}}$ and $\boldsymbol{u}_{l_{(i',2)}}^{\top}\Lambda \boldsymbol{u}_{l_{(i,2)}}$ cancel. (Note: there is a third case where $l_{(i,1)} < l_{(i,2)} < l_{(i',1)} < l_{(i',2)}$ but this case trivially contains no elements of the form $\boldsymbol{u}_{l_{(i,s)}}^{\top}\Lambda \boldsymbol{u}_{l_{(i',s')}}$, $s,s' \in \{0,1\}$.)

We have thus shown that the left-hand side of \Cref{eq:sum_constraints} must in fact be $0$, forming the desired contradiction.
\end{proof}

\section{Proof of Proposition~2}\label{app:prop:Vmap}
For convenience, we restate Propostion~2 here:
\begin{Proposition}\label{prop:app:Vmap}
    Let $V,V'\in i\mathcal{P}^*_k$ for any $k\geq 1$ and $\mathcal{A}$ be an adjoint universal set for $i\mathcal{P}^*_k$. Then, there exist $A_1,...,A_r\in\mathcal{A}$ and $\alpha\in\mathbb{R}^*$ such that $V=\alpha\: ad_{A_1}\cdots ad_{A_{r}}(V')$.
\end{Proposition}

\begin{proof} The proof is established by first demonstrating that for any $V, V' \in i\mathcal{P}^*_k$ there exists a sequence of elements of $P_{1}, ..., P_{t} \in i\mathcal{P}^*_k$, not necessarily in $\mathcal{A}$, such that $V=\alpha'\: ad_{P_1}\cdots ad_{P_{t}}(V')$ with $\alpha' \in \mathbb{R}^{*}$. This is established in Lemma~\ref{lemma:map} below. Since $\mathcal{A}$ is adjoint universal, for each element $P_{i}$ in the sequence, there exist $A^{(i)}_{1},...,A^{(i)}_{s_i}\in\mathcal{A}$ such that, 
\begin{align}
    P_i=\beta_i \:ad_{A^{(i)}_{1}}\cdots ad_{A^{(i)}_{s_i-1}}(A^{(i)}_{s_i})
\end{align} 
with $\beta_{i} \in \mathbb{R}^{*}$ for each $i$. It follows that
\begin{align}
V' = \alpha'' ad_{ad_{A^{(1)}_{1}}\cdots ad_{A^{(1)}_{s_1-1}}(A^{(1)}_{s_1})}\dots ad_{ad_{A^{(t)}_{1}}\cdots ad_{A^{(t)}_{s_t-1}}(A^{(t)}_{s_t})}(V)
\end{align}
for some $\alpha'' \in \mathbb{R}^{*}$. As demonstrated in Lemma~\ref{lemma:reordering}, since all operators involved are elements of $i \mathcal{P}_{N}^{*}$, it is possible to reorder the nested adjoint maps to arrive at a sequence $A_{1}, ..., A_{r}$ such that 
\begin{align}
\alpha\: ad_{A_1}\cdots ad_{A_{r}}(V') = \alpha'' ad_{ad_{A^{(1)}_{1}}\cdots ad_{A^{(1)}_{s_1-1}}(A^{(1)}_{s_1})}\dots ad_{ad_{A^{(t)}_{1}}\cdots ad_{A^{(t)}_{s_t-1}}(A^{(t)}_{s_t})}(V)
\end{align}
with $\alpha \in \mathbb{R}^{*}$.
\end{proof}

In the first Lemma, we prove that we can always find a sequence of Pauli strings that map one Pauli to another under commutation.
\begin{Lemma}\label{lemma:map}
    Let $V,V'\in i\mathcal{P}_k^*$ for any $k\in\mathbb{N}^*$. Then there exist $P_1,...,P_r\in i\mathcal{P}_k^*$ and $\alpha\in\mathbb{R}^*$ such that $V=\alpha\: ad_{P_1}\cdots ad_{P_r}(V')$.
\end{Lemma}
\begin{proof}
    We prove Lemma~\ref{lemma:map} by induction. $N=1$ is trivial. Now consider $V=iV_1\otimes V_2$ and $V'=iV_1'\otimes V_2'$ for $V_1,V_1'\in\mathcal{P}_k$ and $V_2,V_2'\in\mathcal{P}_1$. We consider two cases.

    Case 1: $V_1,V_2,V_1',V_2'\neq I$. By assumption, there exist $P_1\otimes I,...,P_r\otimes I$ such that $V_1\otimes V_2'\propto ad_{P_1}\cdots ad_{P_r}(V_1'\otimes V_2')$. Analogous to the case $N=1$, $V'_2$ can be mapped to an operator proportional to $V_2$.
    
    Case 2: $V_1=I$ or $V_2=I$ or $V_1'=I$ or $V_2'=I$. If $V_1'=I$, then $V_2'\neq I$ and vice verse. If $V_1'=I$, then choose $P_r=A\otimes B$ with $A\neq I$ and $B\neq V_2'$ such that $[P_r,I\otimes V_2']=\alpha\: A\otimes V_2'B$ for $\alpha\in\mathbb{R}^*$. If $V_2'=I$, then choose $P_r=A\otimes B$ with $A\neq V_1'$ and $B\neq I$ such that $[P_r,V_1'\otimes I]=\alpha\: AV_1'\otimes B$ for $\alpha\in\mathbb{R}^*$. If $V_1\neq I$ and $V_2\neq I$, then this reduces to case 1. So from hereon, w.l.o.g. assume that $V'_1,V_2'\neq I$.

    If $V_1=I$, then $V_2\neq I$ and vice verse. If $V_1=I$, then consider case 1 for a target $V=A\otimes V_2B$ such that $A\neq I$ and $[V_2,B]\neq 0$. Then, choose $P_0=A\otimes B$. If $V_2=I$, then consider case 1 for a target $V=V_1A\otimes B$ such that $[V_1,A]\neq 0$ and $B\neq I$. Then, in both cases, $[P_0, V] = \alpha\: V_1\otimes V_2$ for some $\alpha\in\mathbb{R}^*$.
\end{proof}

In the second Lemma we show that the commutators of commutators can be written as nested commutators by a simple re-ordering:
\begin{Lemma}\label{lemma:reordering}
    Let $A,B,C,D\in i\mathcal{P}_{N}^{*}$ and $[[A,B],[C,D]]\neq 0$. Then,
    \begin{align}
        [[A,B],[C,D]]\propto
        \begin{cases}
            ad_D ad_C ad_A(B) \:\:\mathrm{ if }\:\: [AB,D]= 0\\
            ad_C ad_D ad_A(B) \:\:\mathrm{ if }\:\: [AB,C]= 0
        \end{cases}
    \end{align}
\end{Lemma}
\begin{proof}
    We use the binary notation as introduce in Appendix~\ref{app:symplectic_explanation}. Let $\boldsymbol{a}$ be the binary vector corresponding to $A$, and similarly for $\boldsymbol{b},\boldsymbol{c}$ and $\boldsymbol{d}$ for $B,C,D$. By assumption we have that $ad_{ad_{\boldsymbol{a}}(\boldsymbol{b})}(ad_{\boldsymbol{c}}(\boldsymbol{d}))$, which in particular requires that $ad_{\boldsymbol{a}}(\boldsymbol{b}) \neq \boldsymbol{0}$ and $ad_{\boldsymbol{c}}(\boldsymbol{d}) \neq \boldsymbol{0}$. This means that $\boldsymbol{a}^{\top}\Lambda \boldsymbol{b} = \boldsymbol{c}^{\top}\Lambda\boldsymbol{d} = 1$ and also that $ad_{\boldsymbol{a}}(\boldsymbol{b}) = \boldsymbol{a} + \boldsymbol{b}$ and $ad_{\boldsymbol{c}}(\boldsymbol{d}) = \boldsymbol{c} + \boldsymbol{d}$. We can then write
    \begin{align*}
        ad_{ad_{\boldsymbol{a}}(\boldsymbol{b})}(ad_{\boldsymbol{c}} (\boldsymbol{d}))=\left[(\boldsymbol{a}+\boldsymbol{b})^{\top}\Lambda(\boldsymbol{c}+\boldsymbol{d})\right]\left(\boldsymbol{a}+\boldsymbol{b} + \boldsymbol{c}+\boldsymbol{d}\right).
    \end{align*}
    The assumption that the above expression is non-zero also requires both that $\boldsymbol{a} + \boldsymbol{b} + \boldsymbol{c} + \boldsymbol{d} \neq \boldsymbol{0}$ as well as $(\boldsymbol{a}+\boldsymbol{b})^{\top}\Lambda (\boldsymbol{c}+\boldsymbol{d}) = 1$ or equivalently $(\boldsymbol{a}+\boldsymbol{b})^{\top}\Lambda \boldsymbol{c} = 1 - (\boldsymbol{a}+\boldsymbol{b})^{\top}\Lambda \boldsymbol{d}$. There are thus two possibilities:
    \begin{itemize}
        \item If $(\boldsymbol{a}+\boldsymbol{b})^{\top}\Lambda \boldsymbol{d}=0$, then $(\boldsymbol{a}+\boldsymbol{b})^{\top}\Lambda \boldsymbol{c} = 1$ which we can equivalently express as $\boldsymbol{c}^{\top}\Lambda  (\boldsymbol{a}+\boldsymbol{b}) = 1$. Since $\boldsymbol{c}^{\top}\Lambda \boldsymbol{d} = \boldsymbol{d}^{\top}\Lambda \boldsymbol{c} =1$, we have that $\boldsymbol{d}^{\top}\Lambda  (\boldsymbol{a}+\boldsymbol{b}+\boldsymbol{c}) = 1$. It follows that $ad_{\boldsymbol{d}}ad_{\boldsymbol{c}}ad_{\boldsymbol{a}}(\boldsymbol{b})$ is non-zero and equal to $\boldsymbol{a} + \boldsymbol{b} + \boldsymbol{c} + \boldsymbol{d} = ad_{ad_{\boldsymbol{a}}(\boldsymbol{b})}(ad_{\boldsymbol{c}} (\boldsymbol{d}))$, which corresponds to the first case in the statement of the proposition.
        \item If $(\boldsymbol{a}+\boldsymbol{b})^{\top}\Lambda \boldsymbol{c}=0$, then $(\boldsymbol{a}+\boldsymbol{b})^{\top}\Lambda \boldsymbol{d} = 1$. By directly analogous reasoning to the above case with $\boldsymbol{c}$ and $\boldsymbol{d}$ interchanged, we get that $ad_{\boldsymbol{c}}ad_{\boldsymbol{d}}ad_{\boldsymbol{a}}(\boldsymbol{b}) = \boldsymbol{a} + \boldsymbol{b} + \boldsymbol{c} + \boldsymbol{d} = ad_{ad_{\boldsymbol{a}}(\boldsymbol{b})}(ad_{\boldsymbol{c}} (\boldsymbol{d}))$, which is the second case.
    \end{itemize} 
\end{proof}
%Summarizing up to this point, Lemma~\ref{lemma:map} shows Proposition~\ref{prop:Vmap} for $\mathcal{A}=i\mathcal{P}_k$ and Lemma~\ref{lemma:reordering} shows that any commutator of commutators can be written as a nested commutator. That is, we can consider a sequence $P_1,...,P_r\in i\mathcal{P}_k$ such that $V\propto ad_{P_1}\cdots ad_{P_r}(V')$ is nonzero. Since $\langle\mathcal{A}\rangle_{\textrm{Lie}}=\mathfrak{su}(2^k)$, we have that $P_1,...,P_r\in\langle \mathcal{A}\rangle_{[\cdot,\cdot]}$. That is, there exist $A^{(i)}_{1},...,A^{(i)}_{s_i}\in\mathcal{A}$ such that, 
%\begin{align}
%    P_i=\beta_i \:ad_{A^{(i)}_{1}}\cdots ad_{A^{(i)}_{s_i-1}}(A^{(i)}_{s_i})
%\end{align}
%for all $i=1,...,s$ and some $\beta\in\mathbb{R}^*$.
%Consider Lemma~\ref{lemma:reordering} for $[[A,B],[C,V']]=-[[C,V'],[A,B]]$ such that the re-ordered, nested commutator still has $V'$ as last element. From this, we see that we can apply the lemma iteratively to yield a nested commutator of the desired form.

\section{Proof of Theorem~3}\label{appendix:thm:optimal}
For convenience, we restate Theorem~3 here:
\begin{Theorem}\label{thm:app:optimal}
     Consider a set of operators $\mathcal{A}'\cup \mathcal{B}'$ as in Theorem~\ref{thm:app:construction} such that $|\mathcal{A}'\cup \mathcal{B}'|=2N+1$. Let $P\in i\mathcal{P}_N$. Then, \textsc{PauliCompiler}($P$) returns a sequence of operators $G_1,...,G_L\in \mathcal{A}'\cup \mathcal{B}'$ such that $P\propto ad_{G_1}\cdots ad_{G_{L-1}}(G_L)$ and $L=\mathcal{O}(N)$.
\end{Theorem}
\begin{proof}
    Algorithms \textsc{PauliCompiler} and \textsc{SubsystemCompiler} are depicted in Figs.~\ref{fig:app-paulicompiler} and~\ref{fig:app-alg}, respectively. We refer to the line numbering there in the following.
    
    We have proven through Lemmas~\ref{lemma:alg:paulicompiler} and~\ref{lemma:alg}, that \textsc{PauliCompiler} deterministically returns such a sequence $G_1,...,G_L\in\mathcal{A}'\cup \mathcal{B}'$ in finite time.
    
    Now, let us prove $L=\mathcal{O}(N)$.
    With the exception of the case $W=I$ in line 2 of Fig.~\ref{fig:app-paulicompiler}, all returned sequence rely on the output of \textsc{SubsystemCompiler}. The case $W=I$ returns a sequence that is constant in length since $s$ is constant. Therefore, let us consider the length of a sequence $\mathcal{G}$ returned by \textsc{SubsystemCompiler} in Fig.~\ref{fig:app-alg}. 
    It follows immediately from the proof of Lemma~\ref{lemma:alg} in Appendix~\ref{appendix:lemma:alg} that the worst case length of the returned sequence $\mathcal{G}$ is $4(N-k)$. To see this consider $|\mathcal{B}'|=2(N-k)$ and $|\mathcal{A}'|=2k+1$. Then, the worst case length for an optimal choice of a product $\prod_{i=1,...,r}B_i$ to realize another Pauli operator proportional to $iW \in i\mathcal{P}_{N-k}$ is $r=N-k$. This is the case for $W\propto\prod_{B\in\mathcal{B}}B$. Since the set $\mathcal{B}$ only contains independent elements, any other operator in $i\mathcal{P}_{N-k}$ can be written as a product with fewer factors in $\mathcal{B}$ (this can be seen, for example, by taking the view that $\mathcal{B}$ corresponds to a basis of $\mathbb{F}_2^{2N}$ as discussed in Appendix~\ref{app:symplectic_explanation}). 

    The longest worst case sequence of \textsc{PauliCompiler} in Fig.~\ref{fig:app-paulicompiler} is returned in line 11 if both sequences $\mathcal{G}'$ in line 7 and $\mathcal{G}''$ in line 8 have length $4(N-k)$, respectively. The worst case length of the returned sequence is therefore, $L=8(N-k)+s$. Since $s$ is chosen to be a constant in line 9, $L=\mathcal{O}(N)$ in the worst case.
\end{proof}

    Note that we consider the case $|\mathcal{A}'\cup \mathcal{B}'|=2N+1$ because otherwise, we could choose, for example, $\mathcal{B}=i\mathcal{P}_{N-k}$ such that any sequence returned by \textsc{SubsystemCompiler} can be chosen to be of size 1. Further note that this compiler will only return a sequence length of optimal \emph{complexity}. A breadth-first search may still return a sequence that is shorter by a constant factor. As shown in Appendix~\ref{app:time} however, unlike a breadth-first search, \textsc{PauliCompiler} returns a sequence in $\mathcal{O}(N^{2})$ time.

\section{Algorithm~\ref{fig:app-paulicompiler} has $\mathcal{O}(N^2)$ time complexity}\label{app:time}
In this appendix, we show that Algorithm~\ref{fig:app-paulicompiler} is generally preferable over a breadth-first search, even if it returns longer sequences (by a constant), because of its runtime efficiency:
\begin{Proposition}\label{prop:time}
    \textsc{PauliCompiler} (Algorithm~\ref{fig:app-paulicompiler}) has time complexity $\mathcal{O}(N^2)$.
\end{Proposition}
\begin{proof}
    To demonstrate that the algorithm requires $\mathcal{O}(N^2)$ time instead of exponential time, it is convenient to make use of the $2N$-dimensional binary vector representation $\mathbb{F}_2^{2N}$ of Pauli strings and the symplectic form (as described in Appendix~\ref{app:symplectic_explanation}).

    In particular, we show in Lemma~\ref{lemma:subsystem-time} below that the subroutine \textsc{SubsystemCompiler} has quadratic time complexity. 
    
    With this lemma, we consider the three subroutines of \textsc{PauliCompiler} (lines 2-3, 5-12, 14-16) and show that they have quadratic time complexity at worst, considering the binary vector representation of Pauli strings.
\begin{itemize}
    \item \textbf{Lines 2-3:}  In this subroutine, we search for a constant-size subset of a constant-size set that fulfills a certain condition, which can be checked in constant-time. 
    \item \textbf{Lines 5-12:} The choice in line $6$ is equivalent to finding two vectors $\boldsymbol{w}_1,\boldsymbol{w}_2\in\mathbb{F}^{2(N-k)}$ with the condition that $\boldsymbol{w}= \boldsymbol{w}_1^T\Lambda \boldsymbol{w}_2 (\boldsymbol{w}_1+\boldsymbol{w}_2)$ for some other nonzero vector $\boldsymbol{w}\in\mathbb{F}^{2(N-k)}$ in line 6, requires no more than linear time. Lines 7 and 8 have quadratic time complexity due to Lemma~\ref{lemma:subsystem-time}. In line 9, we search for a set of constant size within a set of constant size that fulfills a certain condition which can be checked in constant time. In line 11, we use Lemma~\ref{lemma:reordering} to find a permutation of a list of size $\mathcal{O}(N)$. This step maps a commutator of commutators to a nested commutator. As can be seen from the statement of the lemma, each iteration of the lemma increases the nesting by 1. Since the sequence has linear length, Lemma~\ref{lemma:reordering} needs to be applied at most $\mathcal{O}(N)$ times.
    \item \textbf{Lines 14-16:} Line 14 has quadratic time complexity due to Lemma~\ref{lemma:subsystem-time}. In line 15, we search for a constant-size subset of a constant-size set that fulfills a specified condition, which can be checked in constant-time. 
\end{itemize}
\end{proof}

It just remains to show that the time complexity of the \textsc{SubsystemCompiler} algorithm in Fig.~\ref{fig:app-alg} is quadratic:
\begin{Lemma}\label{lemma:subsystem-time}
    \textsc{SubsystemCompiler} (Algorithm~\ref{fig:app-alg}) has time complexity $\mathcal{O}(N^2)$.
\end{Lemma}
\begin{proof}
    As above, we make use of the binary vector representation of Pauli strings. To find the $r=\mathcal{O}(N)$ operators in line 1, we can make use of the fact that the $\mathcal{B}$ contains a basis of $\mathbb{F}_2^{2(N-k)}$ which we use to find the decomposition of $W$ in $\mathcal{O}(N)$ time steps. To calculate the product of Pauli strings, we have to perform addition modulo one in the symplectic picture which has time complexity $\mathcal{O}(1)$ if we store $\mathcal{O}(N)$ values in memory. The symplectic product $\boldsymbol{v}^T\Lambda\boldsymbol{w}$ between two binary vectors $\boldsymbol{v},\boldsymbol{w}\in \mathbb{F}_2^{2N}$ has, at worst, time complexity of the dot product between two vectors which is $\mathcal{O}(N)$.\\
    
    Next, let us consider the \textsc{while} loop. As we have already shown in the proof of Lemma~\ref{lemma:alg}, for each step $i\leq r=\mathcal{O}(N)$, the  \textsc{while} loop will at most execute 3 different subroutines (lines 6-9, lines 10-13, lines 15-16), each of which runs at worst in linear time as we will see below. This leads to the quadratic time complexity.
    \begin{itemize}
    \item \textbf{Lines 6-9:} To check the \textsc{if} statement in line 6 in the binary vector representation of Pauli strings, we check the symplectic product between two vectors of constant size, which has constant time complexity. The choice in line 7 corresponds to searching for two constant-size operator within a constant-size set (namely $\mathcal{A}\subset i\mathcal{P}_k$ for a constant $k$) and calculating their commutation relation with another operator of constant size. All other lines are of constant time.
    \item \textbf{Lines 10-13:} To check the  \textsc{if} statement in line 10, we have to check the symplectic product between two binary vectors of size $2N$, which has linear time complexity at worst. The choice in line 11 corresponds to searching for a constant-size operator within a constant-size set (namely $\mathcal{A}\subset i\mathcal{P}_k$ for a constant $k$) and calculating its commutation relation with two other operators of constant size. All other lines are of constant time.
    \item \textbf{Lines 15-16:} These lines append an item to a list and update the loop variable and are thus of constant time.
    \end{itemize}
\end{proof}

\section{Detailed construction of Example 3} \label{app:qca_ex}

Here we give the detailed construction of the universal set of Pauli strings described in Example 3.
Consider the following quantum circuit acting on $k$ qubits,
\begin{equation}
    T_k = \prod_{i=1}^k H_iS_i \prod_{i=1}^{k-1} CZ_{i,i+1}
\end{equation}
where $H=(X + Z)/{\sqrt{2}}$, $S=\sqrt{Z}$, and $CZ = (I\otimes I + I\otimes Z + Z\otimes I - Z\otimes Z)/2$ are the usual Hadarmard, phase, and controlled-$Z$ gates, respectively. 
Suppose we have a model of quantum computation that only allows the gates $T_k$ and $e^{i\theta Z_1}$, as in Ref.~\cite{Stephen2024}.
From these gates, we can construct the unitaries $e^{i\theta O_k(\ell)}$, where
\begin{equation}
    O_k(\ell) = T_k^{\ell\dagger} Z_1 T_k^{\ell}
\end{equation}
Since $T_k$ is Clifford, each $O_k(\ell)$ is a Pauli string and there is a minimal $p_k$ such that $T_k^{p_k}\propto I$. This puts us in the scenario described in the main text where gates are generated by Pauli strings from the set $\{ iO_k(\ell) :  l =0,\dots, p_k-1 \}$. For the following, we define the shorthand notation $O_k(-\ell) = O_k(p_k-\ell)$.

From this set, we can construct minimal generating sets of Pauli strings for all $N\in 3\mathbb{N}$. First, by writing $O_k(\ell) = \bigotimes_i O_k(\ell,i)$ where each $O_k(\ell,i)\in \{I,X,Y,Z\}$, we can define the modified operators $\tilde{O}_k(\ell)=\bigotimes_{i\notin 4\mathbb{N}} O_k(\ell,i)$. That is, $\tilde{O}_k(\ell)$ is obtained from $O_k(\ell)$ by simply removing every fourth Pauli from the string. Then, we define the following sets of Pauli strings for each $k\in \mathbb{N}$,
\begin{equation}
    \mathcal{A}_k = \{ \tilde{O}_{4k-1}(\ell): \ell \in \mathcal{I}_{4k-1}\} \,
\end{equation}
where,
\begin{equation}
    \mathcal{I}_{4k-1} = \{-1\}\cup \bigcup_{j=0}^{k-1} \{4j, 4j+1, 4j+2, -4j-2,-4j-3,-4j-4\}\ ,
\end{equation}
such that $|\mathcal{A}_k|=6k+1$. 
Observe that each Pauli string $\tilde{O}_{4k-1}(\ell)$ acts on $N\equiv 3k$ qubits. In order to define the set $\mathcal{A}_k$, it must be the case that $p_{4k-1} \geq 6k+1$ for all $k$, and this is indeed shown to be true in Ref.~\cite{Stephen2024}.

Ref.~\cite{Stephen2024} also showed that $\mathcal{A}_k$ is a universal set of Pauli strings on $3k$ qubits, and that it furthermore has the form described in Theorem \ref{thm:app:construction} where the universal set $\mathcal{A}$ described in the Theorem is defined on the first three qubits (rather than two qubits as in the examples of the main text). Finally, we have $|\mathcal{A}_k|=2N+1$, so these are universal generating sets of minimal size.

\section{Argument for optimal generation rate}

In the main text, we demonstrate that certain universal sets of Pauli strings generate $\mathfrak{su}(2^N)$ faster than others. Here, we argue that the generation rate can be related to the amount of anticommutation between elements of the generating set.
% anticommutation graph of the generating set. To construct this graph from a set of Paulis $\{P_i\}_i$, we define a matrix $G$ such that $G_{i,j} = -1$ if $\{P_i,P_j\}=0$ and $G_{i,j}=1$ otherwise. 
To see how this affects the rate of generation, suppose first that each generator only anticommutes with a small number of other generators. Then, the sets $\mathcal{A}_{ad^{(r)}}$ will be relatively small for small $r$ since most adjoint operations $ad_{P_i}P_j$ are equal to zero. Thus, the growth rate will be very slow to start. This can be observed in Fig. 1 of the main text for the standard generators, since in that case each generator anticommutes with only a constant number $\leq 3$ of other generators. 

On the other hand, suppose that each generator anticommutes with most of the other generators. Then, set $\mathcal{A}_{ad^{(r)}}$ for $r=1$ will be very large since most adjoint operators $ad_{P_i}P_j$ are non-zero. However, the set $\mathcal{A}_{ad^{(r)}}$ for $r=2$ will be very small since most operators $ad_{P_i}ad_{P_j}P_k$ are equal to zero. This even-odd effect will continue for larger values of $r$. Interestingly, for a set of generators that are mutually anticommuting, this same argument shows that $\mathcal{A}_{ad^{(r)}}=\emptyset$ for all $r>1$, meaning that such a set cannot generate all of $\mathfrak{su}(2^N)$.

The above reasoning suggests that an ideal generating set has each generator anticommuting with some intermediate fraction of the other generators. This is indeed the case for the generating set described in Example 3 and Appendix \ref{app:qca_ex} which has each generator anticommuting with approximately half of the others and is accordingly the fastest generating set that we found. Example 1 is an intermediate case where most of the generators anticommute with half of the others, while a few only anticommute with a constant number of others. This could explain why its growth rate is in between that of Example 3 and the standard gate set.

\begin{figure}
    \centering
    \includegraphics[width=0.8\linewidth]{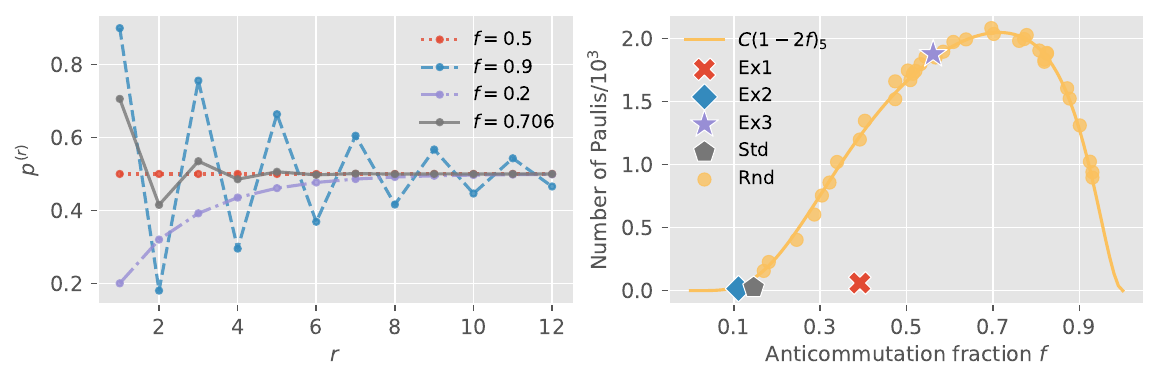}
    \caption{Left: The probability that a given generator anticommutes with an $r$-fold product of generators $p^{(r)}$ as a function of $r$ for different initial anticommutation fractions $f$. Right: Number of Pauli strings generated after five rounds of nested commutators on $N=9$ qubits versus the fraction of pairs of generators that anticommute. Each round point describes a random adjoint universal set of Pauli strings with minimal length $2N+1$. We also show the corresponding points for the various examples described in the main text. The solid line is the $q$-Pochhammer symbol multiplied by a constant $C=16600$ to fit the data by eye.}
    \label{fig:growth_rates}
\end{figure}

% Our claim is that a set of Paulis generates $\mathfrak{su}(2^N)$ optimally fast when approximately half of the entries of $G_{i,j}$ are equal to $-1$, meaning each generator $P_i$ anticommutes with approximately half of the other generators. Let us first examine the two limits. If the Paulis are all mutually commuting, then we cannot construct any new Paulis by taking commutators of the generators, so the final Lie algebra is spanned by the original generating set. That is, $\mathcal{A}_{ad^{(r)}}=\emptyset$ for all $r$. On the other hand, suppose that we have a set of mutually anticommuting Paulis, such that $\{P_i,P_j\} = 0$ for all $i\neq j$. Then, in the first round of generating Paulis, we have $ad_{P_i}P_j \propto P_iP_j$ for all $i\neq j$, so we can add these products to our Lie algebra. However, in the second round, we have $ad_{P_i}P_jP_k = 0$ for all $i,j,k$, so we get no new Paulis and the procedure terminates. That is, $\mathcal{A}_{ad^{(r)}}=\emptyset$ for all $r>1$.

% The two limits can be relaxed to cases where each generator either anticommutes with most of the other generators or each generator commutes with most of others. Then, similar bottlenecks will occur wherein most of the Paulis in $\mathcal{A}_{ad^{(r)}}$ for some level $r$ will equal zero. To avoid these bottlenecks, we therefore want an intermediate case where each generator anticommutes with some fraction of the others.

We can estimate the optimal amount of anticommutation by making the above argument more quantitative. Let us suppose that each generator in a set $\mathcal{A}$ anticommutes with a fraction $f$ of the other generators. In general, this fraction could depend strongly on the generator, but for simplicity we assume $f$ is the same for all generators. When computing $\mathcal{A}_{ad^{(r)}}$, we go through each generator and compute its adjoint with each element in $\mathcal{A}_{ad^{(r-1)}}$. Whenever this adjoint is non-zero, we get a potential new Pauli. The elements of $\mathcal{A}_{ad^{(r-1)}}$ are all products of $r$ generators, so the probability that a given generator anticommutes with an element of $\mathcal{A}_{ad^{(r-1)}}$ is,
\begin{equation}
    p^{(r)} = \sum_{k \text{ odd, } k\leq r} \binom{r}{k} f^k(1-f)^{r-k} = \frac{1-(1-2f)^r}{2}.
\end{equation}
Assuming the elements of $\mathcal{A}_{ad^{(r-1)}}$ are randomly sampled from all $r$-fold products of generators, the expected size of $\mathcal{A}_{ad^{(r)}}$ is,
% \begin{equation} \label{eq:adr}
%     |\mathcal{A}_{ad^{(r)}}| = (|\mathcal{A}| - r)\cdot |\mathcal{A}_{ad^{(r-1)}}|\cdot p^{(r)} = \prod_{k=0}^r (|\mathcal{A}| - k) p^{(k)},
% \end{equation}
\begin{equation}
    |\mathcal{A}_{ad^{(r)}}| =C^{(r)}\cdot \prod_{k=1}^r p^{(k)}
\end{equation}
where the proportionality constant $C^{(r)}$ does not depend on $f$. This constant encodes the fact that not every non-zero commutator will lead to a new Pauli, since it may lead to a Pauli that was already contained in the algebra. This effect, however, is minimal for small $r$, and moreover does not depend on $f$.
% where we defined $p^{(0)}=1$. Then, the total number of Paulis generated after $r$ rounds of nested commutators is,
% \begin{equation}
%     N^{(r)} = |\mathcal{A}| + \sum_{i=1}^r |\mathcal{A}_{ad^{(r)}}| = \sum_{i=0}^r \prod_{k=0}^i (|\mathcal{A}|-k)p^{(k)}
% \end{equation}

To maximize the growth rate for small $r$, we therefore want to maximize the products $\prod_{k=1}^r p^{(r)}$ at each step. The probability $p^{(r)}$ as a function of $r$ has three different regimes of behavior, as shown in Fig.~\ref{fig:growth_rates}. First, when $f=1/2$, we have $p^{(r)}=1/2$ for all $r$. When $f<1/2$, $p^{(r)}$ smoothly increases from $f$ to $1/2$. When $f>1/2$, $p^{(r)}$ oscillates around $1/2$ with decaying amplitude. In all cases, $p^{(r)}$ approaches $1/2$ for large $r$. To determine the optimal value, we notice that the product $\prod_{k=1}^r p^{(r)}$ can be expressed in terms of the so-called $q$-Pochhammer symbol,
\begin{equation} \label{eq:qp}
    (q)_r = \prod_{k=1}^{r} (1-q^k)
\end{equation}
such that,
\begin{equation} \label{eq:qpf}
    \prod_{k=1}^r p^{(r)} = \frac{1}{2^r} (1-2f)_r.
\end{equation}
In the range $(0,1)$, the $q$-Pochhammer symbols $(1-2f)_r$ for increasing values of $r$ quickly converge to their $r\rightarrow\infty$ limit which is a function $\phi(x)=\lim_{r\rightarrow\infty}(x)_r$ called the Euler function. Since the sizes $|\mathcal{A}_{ad^{(r)}}|$ grow exponentially with $r$ (during the early stages of growth), we expect that the dependence of growth rate on $f$ will be governed by the Euler function. 
This function is maximized by the value $x\approx -0.411$ \cite{oeis}, corresponding to $f=f^*\approx 0.706$ in Eq.~\ref{eq:qpf}. This value of $f$ takes advantage of the enlarged value of $p^{(r)}$ for odd $r$ (particularly $r=1$) that occurs when $f>1/2$ without suffering too much from its decreased value for even $r$, as seen in Fig.~\ref{fig:growth_rates}. 

Thus, we conjecture that a generating set of Pauli strings has each generator anticommuting with approximately a fraction $f^*$ of the other generators will generate $\mathfrak{su}(2^N)$ optimally fast. To test this conjecture, we numerically generated random adjoint universal sets of $2N+1$ Pauli operators on $N=9$ qubits and calculated how quickly they generate $\mathfrak{su}(2^N)$ as a function of $f$. Since a randomly chosen set of Pauli strings will typically have $f\approx 0.5$, we modified $f$ by randomly substituting elements of the generating set with new random Pauli strings to drive $f$ towards a target value. In this way, the value of $f$ will be roughly uniform for each Pauli. To determine how quickly each set generates $\mathfrak{su}(2^N)$, we looked at how many Pauli strings had been generated after the fifth round of taking commutators $r=5$ (resulting in Pauli strings that are a product of up to six generators). This number was chosen to be large enough to properly estimate the growth rate, but small enough to stay within the early-time regime to which the above calculation applies (before a large fraction of commutators lead to Pauli strings that are already contained in the algebra). As shown in Fig.~\ref{fig:growth_rates}, the growth rate is in strong agreement with the behavior of the $q$-Pochhammer symbol $(1-2f)_5$.

% We make two remarks on the above argument. First, one can observe that the generating set described in Example 3 and Appendix \ref{app:qca_ex} has each generator anticommuting with approximately half of the others. While this is not the ideal fraction $f^*$, it is still a constant fraction which explains why it generates $\mathfrak{su}(2^N)$ so quickly. In comparison, for the standard standard gate set $\mathcal{A} = \{ X_i,Z_i\}_{i=1,\dots, N}\cup \{Z_i\otimes Z_{i+1}\}_{i=1,\dots,N-1}$, each generator anticommites with only a constant number $\leq 3$ of other generators, hence it generates $\mathfrak{su}(2^N)$ more slowly.

% Note  we did not consider the fine structure of the frustration graph $G_{i,j}$, only its density of $-1$'s. By considering also how these $-1$'s are distributed, it may be possible to gain a deeper understanding on how $\mathfrak{su}(2^N)$ is generated. In this direction, the methods of Ref.~\cite{aguilar2024classificationpauliliealgebras} should be useful.

\end{document}